\RequirePackage[l2tabu, orthodox]{nag}
\documentclass[12pt]{article}
\usepackage{setspace}
\usepackage{makecell}
\newcolumntype{M}[1]{>{\centering\arraybackslash}m{#1}}
\usepackage[T1]{fontenc}

\usepackage{fouriernc}
\usepackage{amsmath}
\usepackage{amsthm}
\usepackage[scale=0.92]{tgschola}
\usepackage[scaled=0.88]{PTSans}
\usepackage{scalefnt,letltxmacro}
\LetLtxMacro{\oldtextsc}{\textsc}
\renewcommand{\textsc}[1]{\oldtextsc{\scalefont{1.25}#1}}
\usepackage{inconsolata}

\usepackage[plain,noend]{algorithm2e}

\usepackage{cleveref}

\DeclareMathOperator*{\argmin}{arg\,min}

\crefname{lemma}{lemma}{lemmas}
\Crefname{lemma}{Lemma}{Lemmas}
\crefname{thm}{theorem}{theorems}
\Crefname{thm}{Theorem}{Theorems}
\crefname{prop}{proposition}{propositions}
\Crefname{prop}{Proposition}{Propositions}
\crefname{algorithm}{algorithm}{algorithms}
\Crefname{algorithm}{Algorithm}{Algorithms}

\newtheorem{thm}{Theorem} 

\newtheorem{lemma}[thm]{Lemma}
\newtheorem{assumption}{Assumption}

\newcommand\independent{\protect\mathpalette{\protect\independenT}{\perp}}
\def\independenT#1#2{\mathrel{\rlap{$#1#2$}\mkern2mu{#1#2}}}

\usepackage[acronym,smallcaps,nowarn]{glossaries}

\usepackage[tmargin=1in, lmargin=1in, rmargin=1in, bmargin=2in]{geometry}
\usepackage{tcolorbox}

\newtheorem{theorem}{Theorem}

\usepackage{times}
\usepackage{bm}
\usepackage{natbib}

\makeatletter
\renewcommand{\algocf@captiontext}[2]{#1\algocf@typo. \AlCapFnt{}#2} 
\def\@algocf@capt@plain{top}
\renewcommand{\algocf@makecaption}[2]{%
  \addtolength{\hsize}{\algomargin}%
  \sbox\@tempboxa{\algocf@captiontext{#1}{#2}}%
  \ifdim\wd\@tempboxa >\hsize
    \hskip .5\algomargin%
    \parbox[t]{\hsize}{\algocf@captiontext{#1}{#2}}
  \else%
    \global\@minipagefalse%
    \hbox to\hsize{\box\@tempboxa}
  \fi%
  \addtolength{\hsize}{-\algomargin}%
}
\makeatother

\usepackage{xcolor}
\usepackage{amsfonts}
\usepackage{amssymb,amsthm}
\usepackage{bbm}
\usepackage{graphicx}
\usepackage{booktabs}
\usepackage{arydshln}
\usepackage{subcaption}
\usepackage{cleveref}

\usepackage{graphics}

\newcommand{\indicator}[1]{\mathbbm{1}_{ {#1} }}

\usepackage[defaultlines=4,all]{nowidow}

\title{\hspace{-0.7cm}\textbf{Minimal Dispersion Approximately Balancing Weights: \\ Asymptotic Properties and Practical Considerations}}

\author{
  Yixin Wang\\
  Department of Statistics\\
  Columbia University\\
  \texttt{yixin.wang@columbia.edu} \\
\\
  Jos\'{e} R. Zubizarreta\\
  Department of Health Care Policy\\
  Department of Statistics\\
  Harvard University\\
  \texttt{zubizarreta@hcp.med.harvard.edu} \\
}

\begin{document}

\maketitle


\begin{abstract} 
Weighting methods are widely used to adjust for covariates in
observational studies, sample surveys, and regression settings. In
this paper, we study a class of recently proposed weighting methods
which find the weights of minimum dispersion that approximately
balance the covariates. We call these weights \emph{minimal weights}
and study them under a common optimization framework. The key
observation is the connection between approximate covariate balance
and shrinkage estimation of the propensity score. This connection
leads to both theoretical and practical developments. From a
theoretical standpoint, we characterize the asymptotic properties of
minimal weights and show that, under standard smoothness conditions on
the propensity score function, minimal weights are consistent
estimates of the true inverse probability weights. Also, we show that
the resulting weighting estimator is consistent, asymptotically
normal, and semiparametrically efficient. From a practical standpoint,
we present a finite sample oracle inequality that bounds the loss
incurred by balancing more functions of the covariates than strictly
needed. This inequality shows that minimal weights implicitly bound
the number of active covariate balance constraints. We finally provide
a tuning algorithm for choosing the degree of approximate balance in
minimal weights. We conclude the paper with four empirical studies
that suggest approximate balance is preferable to exact balance,
especially when there is limited overlap in covariate distributions.
In these studies, we show that the root mean squared error of the
weighting estimator can be reduced by as much as a half with
approximate balance.
\end{abstract}



\section{Introduction}

\subsection{Weighting methods for covariate adjustment}

Weighting methods are widely used to adjust for observed covariates,
for example in observational studies of causal effects
\citep{rosenbaum1987model}, in sample surveys and panel data with unit
non-response \citep{robins1994estimation}, and in regression settings
with missing and/or mismeasured covariates
\citep{hirano2003efficient}. Weighting methods are popular because
they do not require explicitly modeling the outcome
\citep{rosenbaum1987model}. As a result, they are part of the design
stage as opposed to the analysis stage of the study
\citep{rubin2008objective}, which helps to maintain the objectivity of
the study and preserve the validity of its tests
\citep{rosenbaum2010design1}. Furthermore, weighting methods are
considered to be multipurpose in the sense that one set of weights can
be used to estimate the mean of multiple outcomes
\citep{little2014statistical}.

\glsresetall

Conventionally, the weights are estimated by modeling the propensities
of receiving treatment or exhibiting missingness and then inverting
the predicted propensities. However, with this approach it can be
difficult to properly adjust for or balance the observed covariates.
The reason is that this approach only balances covariates in
expectation, by the law of large numbers, but in any particular data
set it can be difficult to balance covariates, especially if the data
set is small or if the covariates are sparse
\citep{zubizarreta2011matching}. In addition, this approach can result
in very unstable estimates when a few observations have very large
weights (e.g., \citealt{kang2007}). To address these problems, a
number of methods have been proposed recently. Instead of explicitly
modeling the propensities of treatment or missingness, these methods
directly balance the covariates. Some of these methods also minimize a
measure of dispersion of the weights. Examples include
\citet{hainmueller2012balancing}, \citet{zubizarreta2015stable},
\citet{chan2016globally}, \citet{zhao2017entropy},
\citet{wong2018kernel}, and \citet{zhao2018covariate}. Earlier and
related methods include \citet{deville1992calibration},
\citet{hellerstein1999imposing}, \citet{imai2014covariate}, and
\citet{li2018balancing}. Two promising methods that use similar
weights together with outcome information are
\citet{athey2018approximate} and \citet{hirshberg2018augmented}. See
\citet{yiu2018covariate} for a framework for constructing weights such
that the association between the covariates and the treatment
assignment is eliminated after weighting.

Most of these weighting methods balance covariates exactly rather than
approximately. This is a subtle but important difference because
approximate balance can trade bias for variance whereas exact balance
cannot. Also, exact balance may not admit a solution whereas
approximate balance may do so. For a fixed sample size, approximate
balance may balance more functions of the covariates than exact
balance.

In this paper, we study the class of weights of minimum dispersion
that approximately balance the covariates. We call these weights
\emph{minimal dispersion approximately balancing weights}, or simply
\emph{minimal weights}. While it has been shown that instances of
minimal weights work well in practice in both low- and
high-dimensional settings (e.g., \citealt{zubizarreta2015stable,
athey2018approximate, hirshberg2018augmented}), and there are valuable
theoretical results (e.g., \citealt{athey2018approximate,
hirshberg2018augmented, wong2018kernel}), important aspects of their
theoretical properties and their practical usage remain to be studied.


\subsection{Theoretical properties and practical considerations of
minimal weights}

In this paper, we study the class of minimal weights. The key
observation is the connection between approximate covariate balance
and shrinkage estimation of the propensity score. This connection
leads to both theoretical and practical developments.

From a theoretical standpoint, we first establish a connection between
minimal weights and shrinkage estimation of the propensity score. We
show that the dual of the minimal weights optimization problem is
similar to parameter estimation in generalized linear models under
$\ell_1$ regularization. This connection allows us to establish the
asymptotic properties of minimal weights by leveraging results on
propensity score estimation. In particular, we show that under
standard smoothness conditions minimal weights are consistent
estimates of the true inverse probability weights both in the $\ell_2$
and $\ell_\infty$ norms.

Next we study the asymptotic properties of a linear estimator based on
minimal weights. We show that the weighting estimator is consistent,
asymptotically normal, and semiparametrically efficient. This result
is related to \citet{chan2016globally}, \citet{fan2016improving},
\citet{zhao2017entropy}, and \citet{zhao2018covariate} in that it
establishes the asymptotic optimality of a similar weighting
estimator. It differs, however, in that it encompasses both
approximate balance and exact balance. The technical conditions
required by this result are among the weakest in the literature: they
are considerably weaker than those required by
\citet{hirano2003efficient} and \citet{chan2016globally}, and are
comparable to those by \citet{fan2016improving}.

From a practical standpoint, we address two problems in minimal
weights: choosing the number of basis functions and selecting the
degree of approximate balance. We derive a finite-sample upper bound
for the potential loss incurred by balancing too many basis functions
of the covariates. This result shows that the loss due to balancing
too many basis functions is hedged by minimal weights because the
number of active balancing constraints is implicitly bounded.

We finally provide a tuning algorithm for calibrating the degree of
approximate balance in minimal weights. This is a general problem in
weighting and thus this algorithm can be of independent interest. We
conclude with four empirical studies that suggest approximate balance
is preferable to exact balance, especially when there is limited
overlap in covariate distributions. These studies show that
approximate balancing weights with the proposed tuning algorithm
yields weighting estimators with considerably lower root mean squared
error than their exact balancing counterparts.


\section{A shrinkage estimation view of minimal weights}

For simplicity of exposition, we focus on the problem of estimating a
population mean from a sample with incomplete outcome data. We assume
the outcomes are missing at random \citep{little2014statistical}.
Under the closely related assumption of strong ignorability
\citep{rosenbaum1983central}, this problem is analogous to estimating
an average treatment effect in an observational study (see
\citealt{kang2007} for an example). See \citet{kang2007} for an
example connecting the problems of causal inference and estimation
with incomplete outcome data.

Consider a random sample of $n$ units from a population of interest,
where some of the units in the sample are missing due to nonresponse.
Let $Z_i$ be the response indicator with $Z_i = 1$ if unit $i$
responds and $Z_i = 0$ otherwise, $i=1, \ldots, n$. Write $r$ for the
total number of respondents. Denote $X_i$ as the (vector of) observed
covariates of unit $i$ and $Y_i$ as the outcome.

Assume there is overlap; that is, the propensity score $\pi(x)
=\mathrm{pr}(Z = 1\mid X = x)$ satisfies $0<\pi(x)<1.$ Furthermore,
assume that the responses are missing at random. This assumption
states that missingness can be fully explained by the observed
covariates: $Y_i \independent Z_i \,|\, X_i$ \citep{robins1997non}.

The goal is to estimate the population mean of the outcome
$\bar{Y}=E(Y_i)$. We use the linear estimator $\hat{Y}_w =
\sum^n_{i=1}w_iZ_iY_i$ for estimation, where the weights $w_i$ adjust
for or balance the observed covariates.

Conventionally, the weights $w_i$ are obtained by fitting a model for
the propensity score $\pi(x)$ and then inverting the predicted
propensities. Despite being widely used, this approach has two
problems in practice: first, balancing the covariates can be difficult
due to misspecification of the propensity score model, if the sample
size is small, or if the covariates are sparse; second, the weighting
estimator can be unstable due to the variability of the weights (see,
e.g., \citealt{zubizarreta2015stable} for a discussion).

To address these problems, several weighting methods have been
proposed recently. These methods are encompassed by the following
mathematical program
\begin{equation}	\label{eq:covbalineq}
\begin{aligned}
& \underset{\boldsymbol{w}}{\text{minimize}}
&& \sum^n_{i=1} Z_i f(w_i)
&&& (1.1) \\
& \text{subject to}
&& \left|\sum^n_{i=1} w_i Z_i B_k(X_i) - \frac{1}{n}\sum^n_{i=1} B_k(X_i) \right| \leq \delta_k, \: k = 1, \ldots, K 
&&& \hspace{0.31cm} (1.2) 
\end{aligned}
\end{equation} 
where $f$ is a convex function of the weights, and $B_k(X_i), k = 1,
\ldots, K$, are smooth functions of the covariates. Typically, the
functions $B_k$ are basis functions for $E(Y_i)$ and are chosen as the
moments of the covariate distributions (see assumptions
\ref{assumption:wtconsistency}.4 and \ref{assumption:wtconsistency}.6
below). Other common choices of $B_k$ include spline
\citep{de1972calculating} and wavelet bases \citep{singh2006optimal}.
The constants $\delta_k$ constrain the imbalances in $B_k$. They are
summarized in the vector $\delta_{K\times 1} = (\delta_1, \ldots,
\delta_K) \geq 0$. In (1.2), we can also constrain the weights to sum
to one, $\sum^n_{i=1} w_i = 1$, and to take positive values, $0 \leq
w_i$, $i = 1, \ldots, n$. These two constraints together ensure that
the weights do not extrapolate; that is, $0 \leq w_i \leq 1, i =
1,\ldots, n$. This is related to the sample boundedness property
discussed by \citet{robins2007comment}, which requires the estimator
to lie within the range of observed values of the outcome.

We call the class of weights that solve the above mathematical program
\emph{minimal dispersion approximately balancing weights}, or simply
\emph{minimal weights}. They have minimal dispersion because they
explicitly minimize a measure of dispersion or extremity of the
weights. They are approximate balancing weights because they have the
flexibility to approximately balance covariates as opposed to exactly.
This flexibility plays an important role in practice by trading bias
for variance.

Special cases of minimal weights are the entropy balancing weights
\citep{hainmueller2012balancing} with $f(x) = x \log x$ and $\delta =
0$, the stable balancing weights \citep{zubizarreta2015stable} with
$f(x) = (x - 1/r)^2$ and $\delta \in \mathbb{R}^+_0$, and the
empirical balancing calibration weights \citep{chan2016globally} with
$f(x) = D(x, 1)$, where $D(x, x_0)$ is a distance measure for a fixed
$x_0\in\mathbb{R}$ that is continuously differentiable in
$x_0\in\mathbb{R}$, non-negative and strictly convex in $x$, and
$\delta = 0$. With the exception of the stable balancing weights,
these methods balance the covariates exactly by letting $\delta = 0$
and assuming the optimization problem is feasible. Related methods
that balance covariates approximately through a Lagrange relaxation of
the balance constraints include \citet{kallus2016generalized},
\citet{athey2018approximate}, \citet{hirshberg2018augmented},
\citet{wong2018kernel}, and \citet{zhao2018covariate}.

The dynamics between the feasibility and the efficacy of covariate
balancing constraints are central to estimation with incomplete
outcome data. Tightening these constraints could make the optimization
program infeasible, but relaxing them could compromise removing biases
due to covariate imbalances.

Studying these dynamics, however, calls for an alternative formulation
of Problem (\ref{eq:covbalineq}) whose solution is easier to
characterize. Theorem \ref{thm:shrink} provides such a formulation. It
writes the dual problem of Problem (\ref{eq:covbalineq}) as an
unconstrained problem by leveraging the structure of minimal weights.
Since Problem (\ref{eq:covbalineq}) is convex, its optimal solution
and the solution to the dual problem will be the same
\citep{boyd2004convex}. Dual formulations of balancing procedures have
been studied by \citet{zhao2017entropy} and \citet{zhao2018covariate}.
\Cref{thm:shrink} helps us to articulate the role of
\emph{approximate} balance constraints.

The dual formulation in \Cref{thm:shrink} establishes a connection
between minimal weights and shrinkage estimation of the propensity
score. At a high level, minimal weights are implicitly fitting a model
for the inverse propensity score with $\ell_1$ regularization; the
model is a generalized linear model on $B_k(\cdot)$, the basis
functions of the covariates.

\begin{theorem}
\label{thm:shrink} The dual of Problem (\ref{eq:covbalineq}) is
equivalent to the unconstrained optimization problem
\begin{align}
\label{eq:dual}
\underset{\lambda}{\text{minimize}} \;
\frac{1}{n}\sum^n_{j=1} \left[-Z_j n\rho \{ B(X_j)^\top \lambda \} +
B(X_j)^\top \lambda\right] + |\lambda|^\top \delta
\end{align} where $\lambda_{K\times 1}$ is the vector of dual
variables associated with the $K$ balancing constraints, and $B(X_j) =
(B_1(X_j), \ldots, B_K(X_j))$ denotes the $K$ basis functions of the
covariates, with $\rho(t) = t/n - t (h')^{-1}(t) + h((h')^{-1}(t))$
and $h(x) = f(1/n - x)$. Moreover, the primal solution $w_j^* $
satisfies
\begin{align}
\label{eq:dual_sol} w_j^* = \rho' \{ B(X_j)^\top \lambda^\dagger \},
\; j = 1, \ldots, n,
\end{align} where $\lambda^\dagger$ is the solution to the dual
optimization problem.
\end{theorem}

The proof is in \Cref{sec:dualproof}. The key to this result is the
form of the constraints (1.2). These box constraints allow us to
eliminate the positivity constraints on the dual variables after a
change of variables.

In Theorem \ref{thm:shrink}, the function $\rho(\cdot)$ is a
transformation of the measure of dispersion of the weights $f(\cdot)$
in (1.1). For example, when $f(x) = x\log x$, as in the entropy
balancing weights \citep{hainmueller2012balancing}, we have $\rho(x) =
-\exp(-x-1)$ and $\rho'(x) = \exp(-x-1)$, which implies a propensity
score model of the form $\pi(x) =
\exp \{ B(x)^\top \lambda + 1 \}$; and when $f(x) = (x-1/r)^2$, as in
the stable balancing weights \citep{zubizarreta2015stable}, we have
$\rho(x) = -x^2/4 + x/r$ and $\rho'(x) = -x/2+1/r$, which implies
$\pi(x) = \{ 1/r - B(x)^\top \lambda/2 \}^{-1}$. At a high level, the
function $\rho'$ can be seen as a link function in generalized linear
models. With specific choices of $\rho'$, \Cref{eq:dual} resembles a
regularized version of the tailored loss function approach in
\citet{zhao2018covariate}.

\Cref{eq:dual} comes down to $\ell_1$ shrinkage estimation. The
inverse propensity score function is estimated as a generalized linear
model on the basis functions $B$ with link function $\rho'$. The dual
variables in $\lambda$ can be seen as the coefficients of the basis
functions in the propensity score regression model. Estimation is
regularized by the weighted $\ell_1$ norm of the coefficients in
$\lambda$. The loss function is
\begin{align}
\label{eq:dualloss} L(\lambda) = -Z n\rho \{ B(x)^\top\lambda \} +
B(x)^\top\lambda.
\end{align} 
The expectation of this loss function is minimized when $\lambda$
satisfies $\{n\pi(x)\}^{-1} = \rho'\{B(x)^\top\lambda\}=w^*$. This is
the key equation connecting minimal weights to the propensity score
$\pi(x)$.

\Cref{thm:shrink} says that if the propensity score depends heavily on
a given covariate, then Problem (\ref{eq:covbalineq}) will try hard to
balance this covariate by assigning it a large dual variable. The dual
variables in $\lambda$ can be interpreted as shadow prices of the
covariate balance constraints (see Section 5.6 of
\citealt{boyd2004convex}). If a constraint has a high shadow price,
then relaxing it by a little will result in a large reduction in the
optimization objective, and vice versa. On the other hand, the
$\ell_1$ penalty decreases the dependence of the weights on covariates
that are hard to balance.

\Cref{thm:shrink} is related to the dual formulation of covariate
balancing scoring rules under regularization
\citep{zhao2018covariate}. The two results have similarities but
differ in their objectives: we use the dual formulation of Problem
(\ref{eq:covbalineq}) to analyze the asymptotic and finite-sample
properties of minimal weights (\Cref{sec:asymptotics} and
\Cref{subsec:oracle}), whereas \citet{zhao2018covariate} uses a
related dual formulation to show that increased regularization in
covariate balancing scoring rules can deteriorate covariate balance.

\section{Asymptotic properties}

\label{sec:asymptotics}

\Cref{thm:shrink} connects minimal weights to shrinkage estimation of
the inverse propensity score function. In this section, we leverage
this connection to characterize the asymptotic properties of minimal
weights. We assume the following conditions hold and prove that
minimal weights are consistent estimates of the inverse propensity
score function $1/\pi(x)$.

\begin{assumption} Assume the following conditions hold:
\vspace*{-6pt}
\label{assumption:wtconsistency}
\begin{enumerate}
\item The minimizer $\lambda^o = \argmin_{\lambda\in \Theta}  E
[-Zn\rho \{ B(X_i)^\top\lambda \} +B(X_i)^\top\lambda ]$ is unique,
where $\Theta$ is the parameter space for $\lambda$.

\item $\lambda^o\in int(\Theta)$, where $\Theta$ is a compact set and
$int(\cdot)$ stands for the interior of a set.

\item There exist a constant $0 < c_0 < 1/2,$ such that $c_0\leq
n\rho'(v) \leq 1 - c_0$ for any $v = B(x)^\top\lambda$ with
$\lambda\in int(\Theta)$. Also, there exist constants $c_1 < c_2 < 0$,
such that $c_1 \leq n\rho''(v) \leq c_2<0$ in some small neighborhood
$\mathcal{B}$ of $v^* = B(x)^\top\lambda^\dagger$.

\item There exists a constant $C$ such that
$\sup_{x\in\mathcal{X}}||B(x)||_2\leq CK^{1/2}$ and
$E\{B(X_i)B(X_i)^\top\}\leq C.$

\item The number of basis functions $K$ satisfies $K = o(n).$

\item There exist constants $r_\pi > 1$ and $\lambda_1^*$ such that
the true propensity score function satisfies
$\sup_{x\in\mathcal{X}}|m^*(x)-B(x)^\top\lambda_1^*| = O(K^{-r_\pi}),$
where $m^*(\cdot) = (\rho')^{-1}[ 1/\{n\pi(x)\} ]$.

\item $||\delta||_2 = O_p \{ K^{1/2}(\log K )/ n + K^{1/2-r_\pi} \}$.
\end{enumerate}
\end{assumption}

Assumptions \ref{assumption:wtconsistency}.1 and
\ref{assumption:wtconsistency}.2 are standard regularity conditions
for consistency of minimum risk estimators. Assumption
\ref{assumption:wtconsistency}.3 enables consistency of
$\lambda^\dagger$ to translate into consistency of the weights. In
particular, the fact that $\rho''$ is bounded implies that the
derivative of the inverse propensity score function is bounded. This
is satisfied by common choices of $f$ in Problem
(\ref{eq:covbalineq}), including the variance, the mean absolute
deviation, and the negative entropy of the weights. Assumption
\ref{assumption:wtconsistency}.4 is a standard technical condition
that restricts the magnitude of the basis functions; see also
Assumption 4.1.6 of \citet{fan2016improving} and Assumption 2(ii) of
\citet{newey1997convergence}. This condition is satisfied by many
classes of basis functions, including the regression spline,
trigonometric polynomial, and wavelet bases
\citep{newey1997convergence, horowitz2004nonparametric, chen2007large,
belloni2015some, fan2016improving}. Assumption
\ref{assumption:wtconsistency}.5 controls the growth rate of the
number of basis functions relative to the number of units. Assumption
\ref{assumption:wtconsistency}.6 is a uniform approximation condition
on the inverse propensity score function. It requires the basis $B(x)$
to be complete, or $m^*(x)$ to be well approximated by a linear model
on $B(x)$. For splines and power series, this assumption is satisfied
by $r_\pi = s / d$, where $s$ is the number of continuous derivatives
of $m^*(\cdot)$ that exist and $d$ is the dimension of $x$ with a
compact domain \citep{newey1997convergence}. Assumption
\ref{assumption:wtconsistency}.7 quantifies the extent to which the
equality covariate balancing constraints can be relaxed such that the
consistency of the resulting weight estimates is maintained.

Under these assumptions, we can prove that minimal weights are
consistent for the inverse propensity score function.

\begin{theorem}
\label{thm:wtconsistency} Let $\lambda^\dagger$ be the solution to
Problem (\ref{eq:covbalineq}) and $w^*(x)= \rho' \{
B(x)^\top\lambda^\dagger \}$. Then, under the conditions in Assumption
\ref{assumption:wtconsistency},
\begin{enumerate}
\item $\sup_{x\in\mathcal{X}} |nw^*(x) -1/\pi(x)|  = O_p \{K(\log K)/n
+ K^{1-r_\pi} \} = o_p(1),$
\item $||nw^*(x)-1/\pi(x)||_{P,2}  = O_p \{K(\log K)/n + K^{1-r_\pi}
\} = o_p(1).$
\end{enumerate}
\end{theorem}

The proof is in \Cref{sec:asymptotics_proof}. It consists of two
steps. First, we show that $\lambda^\dagger$, the solution to the dual
problem, is close to $\lambda_1^*$ in the $\ell_2$ norm. Consistency
of the weights then follows from the Lipschitz property of $\rho'$ and
the bounds on the basis functions in Assumption
\ref{assumption:wtconsistency}. In the special case of exact balance
($\delta=0$), \Cref{thm:wtconsistency} is related to a result in
\citeauthor{fan2016improving} (2016; Appendix D, page 46). This
connection stems from \Cref{thm:shrink}, as minimal weights are
estimating the inverse propensity score.

We now assume the following additional conditions hold and prove that
the resulting weighting estimator is consistent and semiparametrically
efficient for the mean outcome.

\begin{assumption} Assume the following conditions hold:
\vspace*{-6pt}
\label{assumption:regularity}
\begin{enumerate}
  \item $E|Y_i-Y(X_i)|<\infty$, where $Y(x) = E(Y_i|X = x)$.
  \item $E(Y_i^2) < \infty$, where $\bar{Y} = E(Y_i)$ is the
    population mean of the outcome.
  \item There exist $r_y > 1/2$ and $\lambda_2^*$ such that the
  outcome model $Y(x) = E(Y_i|X = x)$ satisfies
  $\sup_{x\in\mathcal{X}}|Y(x)-B(x)^\top\lambda_2^*| = O(K^{-r_y}).$

  \item Let $m^*(\cdot)\in\mathcal{M}$ and $Y(\cdot)\in\mathcal{H},$
  where $m^*(\cdot) = (\rho')^{-1}[ 1/\{n\pi(x)\} ]$ and $Y(\cdot)$ is
  the mean outcome function. $\mathcal{M}$ and $\mathcal{H}$ are two
  sets of smooth functions satisfying $\log n_{[]}\{\varepsilon,
  \mathcal{M}, L_2(P)\}\leq C(1/\varepsilon)^{1/k_1}$ and $\log
  n_{[]}\{\varepsilon,
  \mathcal{H}, L_2(P)\}\leq C(1/\varepsilon)^{1/k_2}$, where $C$ is a
  positive constant and $k_1, k_2 > 1/2$. $n_{[]}\{\varepsilon,
  \mathcal{M}, L_2(P)\}$ denotes the covering number of $\mathcal{M}$
  by $\varepsilon$-brackets.

  \item $n^{0.5(r_\pi+r_y-0.5)^{-1}} = o(K).$
\end{enumerate}
\end{assumption}

Assumptions \ref{assumption:regularity}.1 and
\ref{assumption:regularity}.2 are standard regularity conditions that
ensure that the estimators have finite moments. Assumption
\ref{assumption:regularity}.3 is a uniform approximation condition
similar to Assumption \ref{assumption:wtconsistency}.6 but on the mean
outcome function $Y(x) = E(Y|X=x)$. Assumption
\ref{assumption:regularity}.4 requires that the complexity of the
function classes $\mathcal{M}$ and $\mathcal{H}$ does not increase too
quickly as $\varepsilon$ approaches 0. This assumption is satisfied,
for example, by the H\"{o}lder class with smoothness parameter $s$
defined on a bounded convex subset of $\mathbb{R}^d$ with $s/d > 1/2$
\citep{van1996weak, fan2016improving}; see also Assumption 4.1.7 in
\citet{fan2016improving}. Assumption \ref{assumption:regularity}.5
controls the rate at which $K$ can increase with respect to $n$. In
particular, the rate depends on the sum of $r_\pi$ and $r_y$, which is
the approximation error of the propensity score and the outcome
functions, respectively. This assumption relates to the product
structure of error bounding in doubly robust estimation; see, e.g., Equation (41) of \citet{kennedy2016semiparametric}.

\begin{theorem}
\label{thm:asymptotics} Suppose that the conditions in assumptions
\ref{assumption:wtconsistency} and \ref{assumption:regularity} hold.
Then
\[n^{1/2}(\hat{Y}_{w^*} - \bar{Y}) \stackrel{d}{\rightarrow}
\mathcal{N}(0, V_{opt}),\] where $V_{opt} = \text{var}\{Y(X_i)\} +
E\{\text{var}(Y_i|X_i)/\pi(X_i)\}$ equals the semiparametric efficiency
bound. If in addition $r_y > 1$ holds, then the estimator
\begin{align*}
\begin{split}
\hat{V}_K = &\frac{1}{n}\sum^n_{i=1} \left[ nZ_iw_iY_i -
\sum_{i=1}^nw_iY_i \right.\\ &\left. - B(X_i)^\top \left\{
\frac{1}{n}\sum_{i=1}^nZ_iw_iB(X_i)^\top B(X_i) \right\}^{-1}
\right.
\left.
\left\{ \frac{1}{n}\sum_{i=1}^nZ_iw_iB(X_i)^\top Y_i \right\} (nZ_iw_i
- 1) \right]^2.
\end{split}
\end{align*} is a consistent estimator of the asymptotic variance
$V_{opt}$.
\end{theorem}

The proof is in \Cref{sec:asymptotics_proof}. It uses empirical
process techniques as in \citet{fan2016improving}. The proof involves
the standard decomposition of $\hat{Y}_{w^*} - \bar{Y}$ into four
components, where three of them converge to zero in probability, and
the other one is asymptotically normal and semiparametrically
efficient. Each of the first three components can be controlled by the
bracketing numbers of the function classes to which the inverse
propensity score function and the outcome function belong. Assumption
\ref{assumption:regularity}.2 provides this control.

We conclude this section on asymptotic properties with a discussion on
the uniform approximability assumptions
\ref{assumption:wtconsistency}.6 and \ref{assumption:regularity}.3.
These assumptions depend on both the smoothness of the propensity
score and outcome functions and the dimension $d$ of the covariates.
Suppose both functions belong to the H\"{o}lder class with smoothness
parameter $s$ on the domain $[0,1]^d$. Assumptions
\ref{assumption:wtconsistency}.6 and \ref{assumption:regularity}.3 are
among the weakest in the literature, as they require $s/d > 1$ on the
propensity score function and $s/d > 1/2$ on the outcome. They are
weaker than the assumptions in \citet{hirano2003efficient} which
require $s/d > 7$ on the propensity score function and $s/d > 1$ on
the outcome function, as well as those in \citet{chan2016globally}
which require $s/d>13$ on the propensity score function and $s/d >
3/2$ on the outcome function. They are comparable to those in
\citet{fan2016improving} which require $s/d>1/2$ on the propensity
score function and $s/d > 1/2$ on the outcome function plus the sum of
these two ratios not exceeding $3/2$. To establish these results under
weak assumptions, we use Bernstein's inequality as in
\citet{fan2016improving} and leverage the particular structure of
minimal weights.


\section{Practical considerations}
\label{sec:practice}

\subsection{The loss due to balancing too many functions of the
covariates is bounded}
\label{subsec:oracle}

An important question that arises in practice relates to the cost of
balancing too many basis functions of the covariates. In other words,
practitioners are concerned about how big the loss will be if they
balance more basis functions than needed. This is a valid concern
because \Cref{thm:shrink} implies that, for each basis function $B_k$
we balance, we are implicitly including a similar term in the inverse
propensity score model. Therefore, balancing too many basis functions
could result in estimation loss due to fitting an overly complex
model. The following oracle inequality relieves this concern, as it
shows that this loss is bounded.

\begin{theorem} \label{thm:oracle_words} Let $\lambda^{\dagger}$ be
the solution to the dual of the minimal weights problem
(\ref{eq:dual}) and $\lambda^{\ddagger}$ be the solution to the dual
of the exact balancing weights problem with the number of active
constraints $||\lambda^{\ddagger}||_0$ capped by some constant $C_0 >
0$. Then, under standard technical conditions (see \Cref{sec:oracle}
for details), \begin{align*} E\{L(\lambda^{\dagger}) -
L(\lambda^*_1)\} \leq 3E\{L(\lambda^{\ddagger}) - L(\lambda^*_1)\} +
c_0||\lambda^{\ddagger}||_0, \end{align*} where $\lambda_1^*$ is the
oracle solution as in Assumption \ref{assumption:wtconsistency}.6,
$L(\cdot)$ is the dual loss as in \Cref{eq:dualloss}, and $c_0$ is a
positive constant depending on the number of basis functions $K$.
\end{theorem}

See \Cref{sec:oracle} for technical details. This oracle inequality
bounds $E\{L(\lambda^{\dagger}) - L(\lambda^*_1)\} $, the excess risk
of the minimal weights estimator relative to the oracle estimator
$\lambda_1^*$. We note that the optimal dual loss $L(\lambda)$ is
equal to the optimal primal loss $\sum_{i=1}^nZ_if(w_i)$ (1.1),
because the optimization problem (\ref{eq:covbalineq}) is convex.  A
smaller excess risk translates into a smaller estimation error of the
causal effect estimator.

This inequality compares the linear weighted estimator with two
versions of minimal weights: one with approximate balance, the other
with exact balance. The exact balancing version caps the number of
exact balancing constraints at $C_0$. The inequality shows that the
two estimators have similar risks.

More specifically, when there are few active covariate balancing
constraints, $||\lambda^\ddagger||_0$ will be small. The inequality
then says that the excess risk of approximate balancing in minimal
weights is of the same order as that of exact balancing with its
number of balancing constraints capped. Therefore, balancing
covariates approximately can be seen as implicitly capping the number
of active balancing constraints.

At a high level, this oracle inequality bounds the loss of balancing
too many functions of the covariates with minimal weights.
Fundamentally, the approximate balancing constraints in Problem
(\ref{eq:covbalineq}) are performing $\ell_1$ regularization in the
inverse propensity score estimation problem. This sparse behavior of
the balancing constraints is common in practice; for example, in the
2010 Chilean post-earthquake survey data of
\citeauthor{zubizarreta2015stable} (2015; Figure 1).

\subsection{A tuning algorithm for choosing the degree of approximate
balance $\delta$}

Another practical question that arises with minimal weights is how to
choose the degree of approximate balance $\delta$. In a similar way to
the regularization parameter accompanying the $\ell^1$ norm in lasso
estimation, $\delta$ is a tuning parameter that the investigator needs
to choose. In our setting, choosing $\delta$ is particularly hard;
since there are no outcomes, there is not a clear out-of-sample target
to optimize toward. For choosing $\delta$, we propose \Cref{alg:tune}.

\begin{algorithm}[!h]
\caption{Choosing $\delta$ in minimal weights\label{alg:tune}}
\begin{tabbing}
   \enspace For each $\delta$ in a grid $\mathcal{D}\subset [0, K^{-1/2}]$ of candidate
   imbalances   \\
   \qquad Compute $\{w_i\}_{i=1}^n$ by solving Problem
   (\ref{eq:covbalineq}) \\
   \qquad For each $k \in \{1, ..., K\}$ \\
   \qquad \qquad Draw a bootstrap sample $\mathcal{K}_k$ from the
   original data\\
   \qquad \qquad Evaluate covariate balance $C_k$ on the sample
   $\mathcal{K}_k$, \\
  \qquad \qquad \qquad$C_k:=||\{\sum_{i\in\mathcal{K}_k} w_i Z_i
   B_k(X_i)\}/(\sum_{i\in\mathcal{K}_k} w_i
   Z_i)-\sum^n_{i=1}B_k(X_i)/n||_2/\text{sd}\{B_k(X)\}$\\
   \qquad Compute the mean covariate balance, $C_S(\delta):=
\sum_{k=1}^KC_k/K$\\
\enspace Output $\delta^* = \argmin_{\delta\in
\mathcal{D}}C_S(\delta)$
\end{tabbing}
\end{algorithm}

\glsresetall The key idea behind \Cref{alg:tune} is to use the
covariate balance in the bootstrapped samples as a proxy for how well
the target parameters are estimated. The intuition is that in theory
the true inverse propensity score weights will balance the population
as well as \emph{samples} from this population. Therefore, if the
weights are well-calibrated and robust to sampling variation, they
will have this same property. To this end, we evaluate the covariate
balance on bootstrapped samples $C_S$ with the weights computed from
the original data set. In the following section, we show that the
value of $\delta$ selected by \Cref{alg:tune} often coincides with or
neighbors the optimal $\delta$ that gives the smallest root mean
squared error in estimating the target parameters. We recommend
choosing values of $\delta$ smaller than $K^{-1/2}$ because larger
values are likely to break the conditions in Assumption
\ref{assumption:wtconsistency}.


\subsection{Empirical studies}

We illustrate the performance of minimal weights in four empirical
studies. In these four studies we set $\delta$ with \Cref{alg:tune}
and consider three dispersion measures of the weights: the sum of
absolute deviations, $f(w)=|w-\bar{w}|$, the variance, $f(w) =
(w-1/r)^2$ \citep{zubizarreta2015stable}, and the negative entropy,
$f(w) = w\log w$ \citep{hainmueller2012balancing}. We find that
minimal weights with approximate balance admit a solution in cases
where exact balance does not. Approximate balancing also achieves
considerably lower root mean squared than exact balancing when there
is limited overlap in covariate distributions.

We defer three of the simulation studies to \Cref{sec:empirical_supp}:
one on the \citet{kang2007} example, one on the
\citet{lalonde1986evaluating} data set, and another on the
\citet{wong2018kernel} simulation. Here we present one simulation
study based on the right heart catheterization data set of
\citet{connors1996effectiveness}.

The right heart catheterization data set was first used to study the
effectiveness of right heart catheterization in the initial care of
critically ill patients. The data set has 2998 observations and 77
variables, including covariates, a treatment indicator, and the
outcome. Balancing the 75 available covariates exactly is not feasible
in most of the simulated data sets, so for comparison purposes we
restrict the analyses to the 23 covariates listed in Table 1 of
\citet{connors1996effectiveness}. We generate the data sets and
calculate the minimal weights (both with exact and approximate
balance) using only these 23 covariates.

Based on this data set, we generate 1000 simulated data sets as
follows. We construct the treatment indicator $Z_i$ as $Z_i =
\indicator{ \{ Z^*_i > 0 \} }$ where $Z^*_i = (\alpha + \beta X_i) / c
+ \text{Unif}(-0.5, 0.5)$ and $X_i$ are the observed covariates. In
  the model for $Z^*_i$, $\alpha$ and $\beta$ are obtained by fitting
  a logistic regression to the original treatment indicator in the
  original data set. We simulate two scenarios, one with good overlap
  $(c = 10)$ and another with bad overlap $(c = 1)$. For both
  scenarios, we generate pairs of potential outcomes $\{Y_i(0),
  Y_i(1)\}$ by fitting a regression model to the original treated and control outcomes, and predicting on the
  entire sample. We obtain the observed outcome by letting $Y_i =
  Z_iY_i(1)+(1-Z_i)Y_i(0)$.

In both scenarios, we compare the root mean squared error of the
estimated average treatment effects on both the entire and treated
populations, using both minimal weights with \Cref{alg:tune} and
minimal weights with exact balance (i.e., with $\delta=0$). The
results are presented in \Cref{fig:rhc_ate_fig} and \Cref{table:rhc}.

\begin{table}
\begin{subtable}[]{.5\linewidth}
\begin{tabular}{lrrrr}
\toprule
\multicolumn{1}{c}{} & \multicolumn{2}{c}{Good Overlap}
&\multicolumn{2}{c}{Bad Overlap}\\
\multicolumn{1}{c}{Dispersion} & \multicolumn{1}{c}{Exact} &
\multicolumn{1}{c}{Apprx.} & \multicolumn{1}{c}{Exact} &
\multicolumn{1}{c}{Apprx.}\\
\midrule Abs. Dev.& 0.19 & \textbf{0.18} &  - &
\textbf{0.27}\\ Variance& \textbf{0.16} & 0.17 &  -  &
\textbf{0.26}\\ Neg. Ent.& \textbf{0.16} & \textbf{0.16}&  - &
\textbf{0.27} 
\\
\bottomrule
\end{tabular}
\caption{Average treatment effect
\label{table:rhc_ate}}
\end{subtable}
\begin{subtable}[]{.5\linewidth}
\begin{tabular}{lrrrr}
\toprule
\multicolumn{1}{c}{} & \multicolumn{2}{c}{Good Overlap}
&\multicolumn{2}{c}{Bad Overlap}\\
\multicolumn{1}{c}{Dispersion} & \multicolumn{1}{c}{Exact} &
\multicolumn{1}{c}{Apprx.} & \multicolumn{1}{c}{Exact} &
\multicolumn{1}{c}{Apprx.}\\
\midrule Abs. Dev. & \textbf{0.10} & \textbf{0.10} &  0.24 &
\textbf{0.08}\\ Variance& \textbf{0.09} & \textbf{0.09} &  0.18  &
\textbf{0.07}\\ Neg. Ent. & 0.10 & \textbf{0.09}& 0.20 &
\textbf{0.10} \\
\bottomrule
\end{tabular}
\caption{Average treatment effect on the treated
\label{table:rhc_att}}
\end{subtable}
\caption{Root mean squared error for (a) the average treatment effect
and (b) the average treatment effect on the treated. We bold the
lowest errors for each measure of dispersion. The symbol ``-''
indicates that exact balancing does not admit a solution. In the case
of bad overlap, balancing covariates approximately reduces the error
of the average treatment effect on the treated by a half compared to
exact balance.
\label{table:rhc}}
\end{table}

\begin{figure*}[h!]
    \centering
    \begin{subfigure}[t]{0.5\textwidth}
        \centering
        \includegraphics[width=\linewidth]{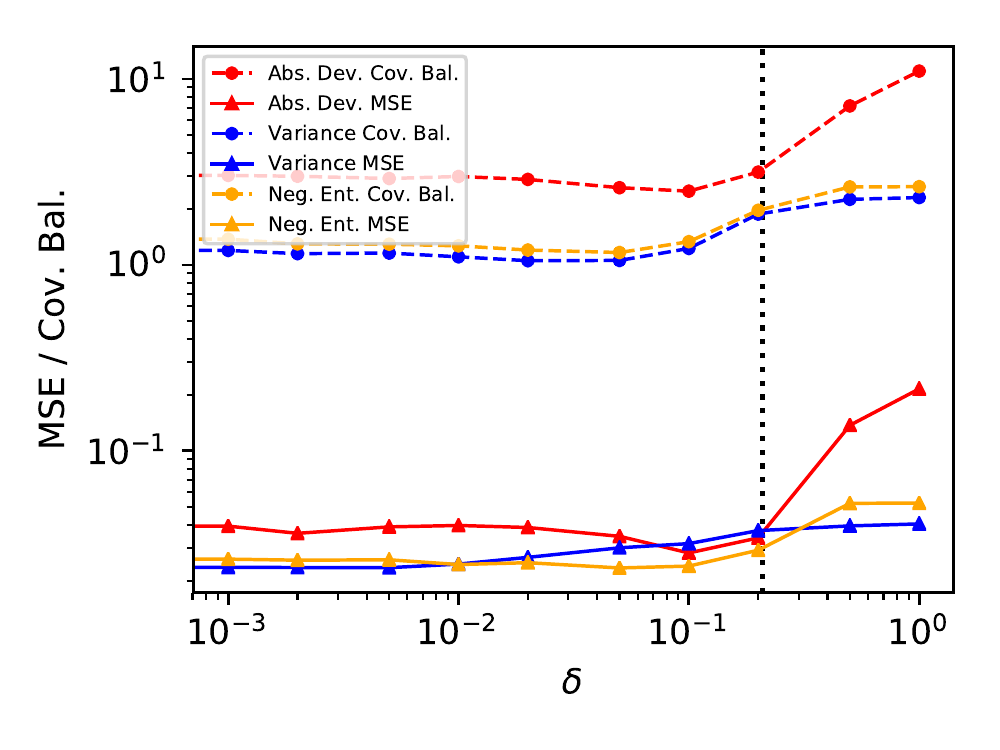}
        \caption{Good overlap, average treatment effect}
    \end{subfigure}%
    \begin{subfigure}[t]{0.5\textwidth}
        \centering
        \includegraphics[width=\linewidth]{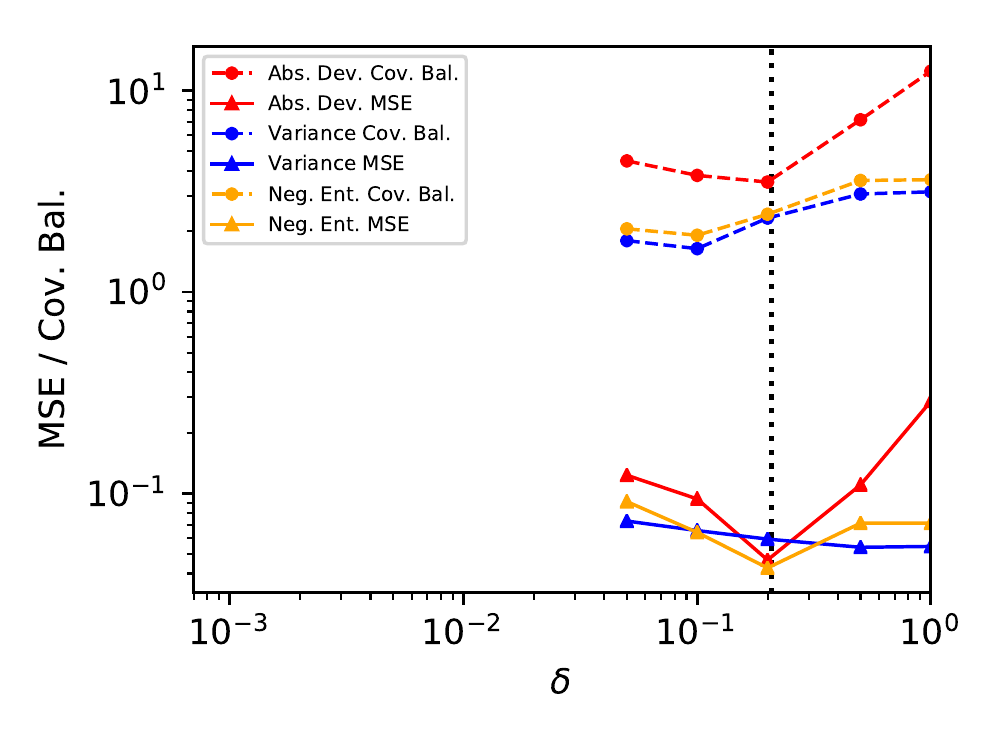}
        \caption{Bad overlap, average treatment effect}
    \end{subfigure}
        \begin{subfigure}[t]{0.5\textwidth}
        \centering
        \includegraphics[width=\linewidth]{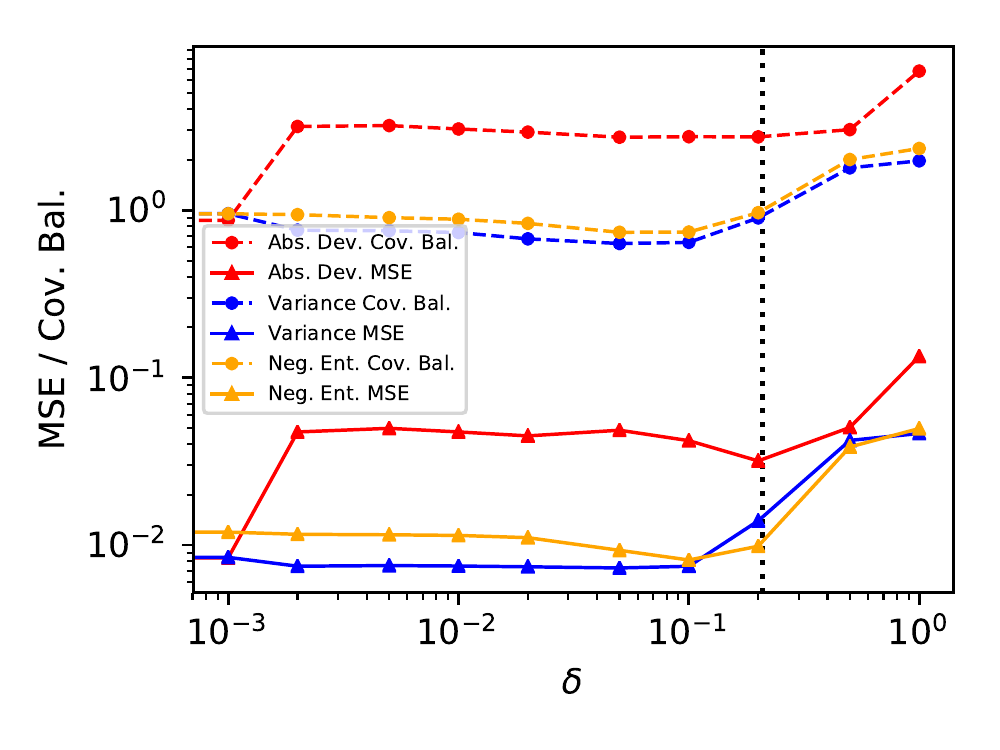}
        \caption{Good overlap, average treatment effect on the treated}
    \end{subfigure}%
    \begin{subfigure}[t]{0.5\textwidth}
        \centering
        \includegraphics[width=\linewidth]{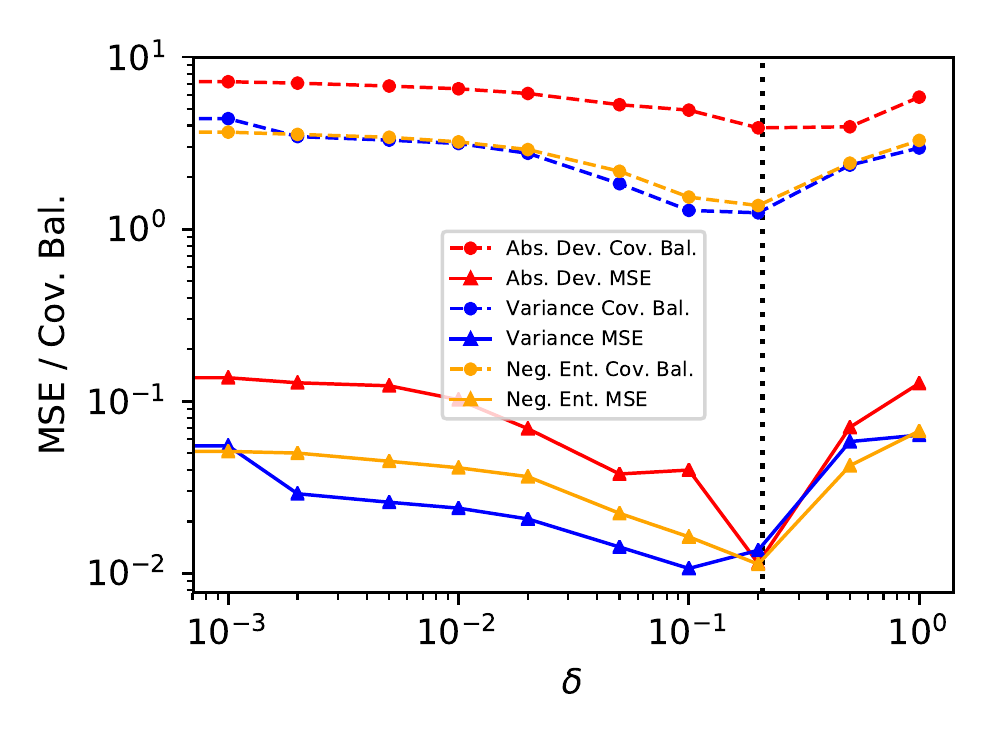}
        \caption{Bad overlap, average treatment effect on the treated}
    \end{subfigure}
    \caption{Mean squared error and bootstrapped covariate balance for
different values of the tuning parameter $\delta$. In the horizontal
axis, $\delta$ starts at 0. The vertical dotted line indicates $\delta
= K^{-1/2}$, where $K$ is the number of basis functions balanced.
Selecting $\delta$ according to the bootstrapped covariate balance, as
in \Cref{alg:tune}, often coincides with or neighbors the optimal
$\delta$ with the smallest error. We recommend choosing values of
$\delta$ smaller than $K^{-1/2}$ as greater values are likely to break
the conditions in Assumption \ref{assumption:wtconsistency}.
\label{fig:rhc_ate_fig}}
\end{figure*}

\Cref{table:rhc}(a) presents the root mean squared error of minimal
weights in estimating the average treatment effect. When the data
exhibits bad overlap, minimal weights provide good estimates whereas
their exact balancing counterpart does not admit a solution. With good
overlap, minimal weights with approximate balancing performs similarly
to exact balancing.

\Cref{table:rhc}(b) shows the results for the average treatment effect
on the treated. In this case, both exact and approximate balance admit
solutions under bad overlap. The table shows that approximate balance
can markedly reduce the root mean squared error relative to exact
balance. We also note that, while we are in a low-dimensional regime
(we balance fewer basis functions than the total number of
observations), approximate balance (or $\ell_1$-regularization) still
helps to reduce the error. The reason is that approximate balance
trades bias for variance. In fact, when there is bad overlap,
traditional weighting estimators that use weights that balance
covariates exactly tend to have high variance as they rely heavily on
a few observations. In such cases, approximate balance can ``pull
back'' from those observations and trade bias for variance to reduce
the overall error.

Figure \ref{fig:rhc_ate_fig} shows that the root mean squared error of
the effect estimates is sensitive to the choice of $\delta$. Moreover,
the value of $\delta$ selected by \Cref{alg:tune} often coincides with
the optimal value of $\delta$ that produces the lowest mean squared
error (solid lines in Figure \ref{fig:rhc_ate_fig}). Again,
\Cref{alg:tune} selects the value of $\delta$ that minimizes the
bootstrapped covariate balance (dashed lines in Figure
\ref{fig:rhc_ate_fig}). We observe that when $\delta$ achieves the
lowest bootstrapped covariate balance (dashed lines) it also reaches
the lowest error (solid lines). In the figure, the dotted line
indicates a value of $\delta$ equal to $K^{-1/2}$, where $K$ is the
number of basis functions of the covariates being balanced. We
recommend choosing values of $\delta$ smaller than $K^{-1/2}$ for
Assumption 1.7 required by \Cref{thm:asymptotics} to hold.

In general, minimal weights tuned with  \Cref{alg:tune} exhibit better
empirical performance in the right heart catheterization data set than
their exact balancing counterparts. Empirical studies with the
\citet{kang2007} example, the \citet{lalonde1986evaluating} data set,
and the \citet{wong2018kernel} simulation exhibit a similar pattern.
See \Cref{sec:empirical_supp} for details.

\section{Summary and remarks}
\glsresetall

Minimal dispersion approximately balancing weights, abbreviated as
\emph{minimal weights}, are the weights of minimal dispersion that
approximately balance covariates. In this paper, we study the class of
minimal weights from theoretical and practical standpoints. From a
theoretical standpoint, we show that under standard technical
assumptions minimal weights are consistent estimates of the true
inverse probability weights. Also, we show that the resulting minimal
weights linear estimator is consistent, asymptotically normal, and
semiparametrically efficient. From a practical standpoint, we derive
an oracle inequality that bounds the loss incurred by balancing too
many functions of the covariates in finite samples. Also, we propose a
tuning algorithm to select the degree of approximate balance in
minimal weights, which can be of independent interest. Finally, we
show that approximate balance is preferable to exact balance in
empirical studies, especially when there is limited overlap in
covariate distributions.

The theoretical results developed in this work can be extended to
matching, where covariates are balanced approximately but with weights
that encode an assignment between matched units (e.g.,
\citealt{rubin1973matching,rosenbaum1989optimal,
hansen2004full,abadie2006large,zubizarreta2012using,diamond2013genetic}).
The tuning algorithm used to select the degree of approximate balance
can also be extended to matching. Promising directions for future work
include doubly robust estimation \citep{robins1995semiparametric}
where propensity score modeling weights can be substituted by minimal
weights (see \citet{athey2018approximate} and
\citet{hirshberg2018augmented}). Also, minimal weights can be extended
to instrumental variables and regression discontinuity settings where
model-based inverse probability weights are used for covariate
adjustments.

\section*{Acknowledgment}

We thank the editor, the associate editor, and two anonymous reviewers
for their insightful comments. We thank David Blei, Zach Branson,
Xinkun Nie, Stefan Wager, Anna Zink, and Qingyuan Zhao for their
valuable feedback on our manuscript. We also thank the support from
the Alfred P. Sloan Foundation.

\bibliographystyle{apa}
\bibliography{mybibliography5}

\begin{thebibliography}{}

\bibitem[\protect\astroncite{Abadie and Imbens}{2006}]{abadie2006large}
Abadie, A. and Imbens, G.~W. (2006).
\newblock Large sample properties of matching estimators for average treatment
  effects.
\newblock {\em Econometrica}, 74(1):235--267.

\bibitem[\protect\astroncite{Athey et~al.}{2018}]{athey2018approximate}
Athey, S., Imbens, G.~W., and Wager, S. (2018).
\newblock Approximate residual balancing: De-biased inference of average
  treatment effects in high dimensions.
\newblock {\em Journal of the Royal Statistical Society: Series B}.

\bibitem[\protect\astroncite{Belloni et~al.}{2015}]{belloni2015some}
Belloni, A., Chernozhukov, V., Chetverikov, D., and Kato, K. (2015).
\newblock Some new asymptotic theory for least squares series: Pointwise and
  uniform results.
\newblock {\em Journal of Econometrics}, 186(2):345--366.

\bibitem[\protect\astroncite{Boyd and Vandenberghe}{2004}]{boyd2004convex}
Boyd, S. and Vandenberghe, L. (2004).
\newblock {\em {Convex Optimization}}.
\newblock Cambridge University Press.

\bibitem[\protect\astroncite{Chan et~al.}{2016}]{chan2016globally}
Chan, K. C.~G., Yam, S. C.~P., and Zhang, Z. (2016).
\newblock Globally efficient nonparametric inference of average treatment
  effects by empirical balancing calibration weighting.
\newblock {\em Journal of the Royal Statistical Society: Series B},
  78(3):673--700.

\bibitem[\protect\astroncite{Chen}{2007}]{chen2007large}
Chen, X. (2007).
\newblock Large sample sieve estimation of semi-nonparametric models.
\newblock {\em Handbook of Econometrics}, 6:5549--5632.

\bibitem[\protect\astroncite{Connors et~al.}{1996}]{connors1996effectiveness}
Connors, A.~F., Speroff, T., Dawson, N.~V., Thomas, C., Harrell, F.~E., Wagner,
  D., Desbiens, N., Goldman, L., Wu, A.~W., Califf, R.~M., et~al. (1996).
\newblock The effectiveness of right heart catheterization in the initial care
  of critically ill patients.
\newblock {\em Journal of the American Medical Association}, 276(11):889--897.

\bibitem[\protect\astroncite{De~Boor}{1972}]{de1972calculating}
De~Boor, C. (1972).
\newblock On calculating with b-splines.
\newblock {\em Journal of Approximation Theory}, 6(1):50--62.

\bibitem[\protect\astroncite{Deville and
  S{\"a}rndal}{1992}]{deville1992calibration}
Deville, J.-C. and S{\"a}rndal, C.-E. (1992).
\newblock Calibration estimators in survey sampling.
\newblock {\em Journal of the American Statistical Association},
  87(418):376--382.

\bibitem[\protect\astroncite{Diamond and Sekhon}{2013}]{diamond2013genetic}
Diamond, A. and Sekhon, J.~S. (2013).
\newblock Genetic matching for estimating causal effects: A general
  multivariate matching method for achieving balance in observational studies.
\newblock {\em Review of Economics and Statistics}, 95(3):932--945.

\bibitem[\protect\astroncite{Fan et~al.}{2016}]{fan2016improving}
Fan, J., Imai, K., Liu, H., Ning, Y., and Yang, X. (2016).
\newblock Improving covariate balancing propensity score: A doubly robust and
  efficient approach.

\bibitem[\protect\astroncite{Hahn}{1998}]{hahn1998role}
Hahn, J. (1998).
\newblock On the role of the propensity score in efficient semiparametric
  estimation of average treatment effects.
\newblock {\em Econometrica}, pages 315--331.

\bibitem[\protect\astroncite{Hainmueller}{2012}]{hainmueller2012balancing}
Hainmueller, J. (2012).
\newblock Entropy balancing for causal effects: a multivariate reweighting
  method to produce balanced samples in observational studies.
\newblock {\em Political Analysis}, 20(1):25--46.

\bibitem[\protect\astroncite{Hansen}{2004}]{hansen2004full}
Hansen, B.~B. (2004).
\newblock Full matching in an observational study of coaching for the {SAT}.
\newblock {\em Journal of the American Statistical Association},
  99(467):609--618.

\bibitem[\protect\astroncite{Hellerstein and
  Imbens}{1999}]{hellerstein1999imposing}
Hellerstein, J.~K. and Imbens, G.~W. (1999).
\newblock Imposing moment restrictions from auxiliary data by weighting.
\newblock {\em Review of Economics and Statistics}, 81(1):1--14.

\bibitem[\protect\astroncite{Hirano et~al.}{2003}]{hirano2003efficient}
Hirano, K., Imbens, G.~W., and Ridder, G. (2003).
\newblock Efficient estimation of average treatment effects using the estimated
  propensity score.
\newblock {\em Econometrica}, 71(4):1161--1189.

\bibitem[\protect\astroncite{Hirshberg and
  Wager}{2018}]{hirshberg2018augmented}
Hirshberg, D.~A. and Wager, S. (2018).
\newblock Augmented minimax linear estimation.
\newblock {\em arXiv:1712.00038}.

\bibitem[\protect\astroncite{Horowitz et~al.}{2004}]{horowitz2004nonparametric}
Horowitz, J.~L., Mammen, E., et~al. (2004).
\newblock Nonparametric estimation of an additive model with a link function.
\newblock {\em Annals of Statistics}, 32(6):2412--2443.

\bibitem[\protect\astroncite{Imai and Ratkovic}{2014}]{imai2014covariate}
Imai, K. and Ratkovic, M. (2014).
\newblock Covariate balancing propensity score.
\newblock {\em Journal of the Royal Statistical Society: Series B},
  76(1):243--263.

\bibitem[\protect\astroncite{Kallus}{2016}]{kallus2016generalized}
Kallus, N. (2016).
\newblock Generalized optimal matching methods for causal inference.
\newblock {\em arXiv:1612.08321}.

\bibitem[\protect\astroncite{Kang and Schafer}{2007}]{kang2007}
Kang, J. D.~Y. and Schafer, J.~L. (2007).
\newblock {D}emystifying double robustness: A comparison of alternative
  strategies for estimating a population mean from incomplete data (with
  discussion).
\newblock {\em Statistical Science}, 22(4):523--539.

\bibitem[\protect\astroncite{Kennedy}{2016}]{kennedy2016semiparametric}
Kennedy, E.~H. (2016).
\newblock Semiparametric theory and empirical processes in causal inference.
\newblock In {\em Statistical Causal Inferences and Their Applications in
  Public Health Research}, pages 141--167. Springer.

\bibitem[\protect\astroncite{LaLonde}{1986}]{lalonde1986evaluating}
LaLonde, R.~J. (1986).
\newblock Evaluating the econometric evaluations of training programs with
  experimental data.
\newblock {\em The American Economic Review}, pages 604--620.

\bibitem[\protect\astroncite{Li et~al.}{2018}]{li2018balancing}
Li, F., Morgan, K.~L., and Zaslavsky, A.~M. (2018).
\newblock Balancing covariates via propensity score weighting.
\newblock {\em Journal of the American Statistical Association},
  113(521):390--400.

\bibitem[\protect\astroncite{Little and Rubin}{2014}]{little2014statistical}
Little, R.~J. and Rubin, D.~B. (2014).
\newblock {\em Statistical analysis with missing data}.
\newblock John Wiley \& Sons.

\bibitem[\protect\astroncite{Newey}{1997}]{newey1997convergence}
Newey, W.~K. (1997).
\newblock Convergence rates and asymptotic normality for series estimators.
\newblock {\em Journal of Econometrics}, 79(1):147--168.

\bibitem[\protect\astroncite{Robins et~al.}{2007}]{robins2007comment}
Robins, J., Sued, M., Lei-Gomez, Q., and Rotnitzky, A. (2007).
\newblock Comment: Performance of double-robust estimators when" inverse
  probability" weights are highly variable.
\newblock {\em Statistical Science}, 22(4):544--559.

\bibitem[\protect\astroncite{Robins and Gill}{1997}]{robins1997non}
Robins, J.~M. and Gill, R.~D. (1997).
\newblock Non-response models for the analysis of non-monotone ignorable
  missing data.
\newblock {\em Statistics in medicine}, 16(1):39--56.

\bibitem[\protect\astroncite{Robins and
  Rotnitzky}{1995}]{robins1995semiparametric}
Robins, J.~M. and Rotnitzky, A. (1995).
\newblock Semiparametric efficiency in multivariate regression models with
  missing data.
\newblock {\em Journal of the American Statistical Association},
  90(429):122--129.

\bibitem[\protect\astroncite{Robins et~al.}{1994}]{robins1994estimation}
Robins, J.~M., Rotnitzky, A., and Zhao, L.~P. (1994).
\newblock Estimation of regression coefficients when some regressors are not
  always observed.
\newblock {\em Journal of the American Statistical Association},
  89(427):846--866.

\bibitem[\protect\astroncite{Rosenbaum}{1987}]{rosenbaum1987model}
Rosenbaum, P.~R. (1987).
\newblock Model-based direct adjustment.
\newblock {\em Journal of the American Statistical Association},
  82(398):387--394.

\bibitem[\protect\astroncite{Rosenbaum}{1989}]{rosenbaum1989optimal}
Rosenbaum, P.~R. (1989).
\newblock Optimal matching for observational studies.
\newblock {\em Journal of the American Statistical Association}, 84:1024--1032.

\bibitem[\protect\astroncite{Rosenbaum}{2010}]{rosenbaum2010design1}
Rosenbaum, P.~R. (2010).
\newblock {\em Design of Observational Studies}.
\newblock Springer.

\bibitem[\protect\astroncite{Rosenbaum and Rubin}{1983}]{rosenbaum1983central}
Rosenbaum, P.~R. and Rubin, D.~B. (1983).
\newblock The central role of the propensity score in observational studies for
  causal effects.
\newblock {\em Biometrika}, 70(1):41--55.

\bibitem[\protect\astroncite{Rubin}{1973}]{rubin1973matching}
Rubin, D.~B. (1973).
\newblock Matching to remove bias in observational studies.
\newblock {\em Biometrics}, 29:159--183.

\bibitem[\protect\astroncite{Rubin}{2008}]{rubin2008objective}
Rubin, D.~B. (2008).
\newblock For objective causal inference, design trumps analysis.
\newblock {\em Annals of Applied Statistics}, 2(3):808--840.

\bibitem[\protect\astroncite{Singh and Tiwari}{2006}]{singh2006optimal}
Singh, B.~N. and Tiwari, A.~K. (2006).
\newblock Optimal selection of wavelet basis function applied to ecg signal
  denoising.
\newblock {\em Digital signal processing}, 16(3):275--287.

\bibitem[\protect\astroncite{Tropp et~al.}{2015}]{tropp2015introduction}
Tropp, J.~A. et~al. (2015).
\newblock An introduction to matrix concentration inequalities.
\newblock {\em Foundations and Trends{\textregistered} in Machine Learning},
  8(1-2):1--230.

\bibitem[\protect\astroncite{Tseng and Bertsekas}{1987}]{tseng1987relaxation}
Tseng, P. and Bertsekas, D.~P. (1987).
\newblock Relaxation methods for problems with strictly convex separable costs
  and linear constraints.
\newblock {\em Mathematical Programming}, 38(3):303--321.

\bibitem[\protect\astroncite{Tseng and Bertsekas}{1991}]{tseng1991relaxation}
Tseng, P. and Bertsekas, D.~P. (1991).
\newblock Relaxation methods for problems with strictly convex costs and linear
  constraints.
\newblock {\em Mathematics of Operations Research}, 16(3):462--481.

\bibitem[\protect\astroncite{Van~de Geer}{2008}]{van2008high}
Van~de Geer, S.~A. (2008).
\newblock High-dimensional generalized linear models and the lasso.
\newblock {\em Annals of Statistics}, pages 614--645.

\bibitem[\protect\astroncite{Van Der~Vaart and Wellner}{1996}]{van1996weak}
Van Der~Vaart, A.~W. and Wellner, J.~A. (1996).
\newblock Weak convergence.
\newblock In {\em Weak Convergence and Empirical Processes}, pages 16--28.
  Springer.

\bibitem[\protect\astroncite{Wong and Chan}{2018}]{wong2018kernel}
Wong, R.~K. and Chan, K. C.~G. (2018).
\newblock Kernel-based covariate functional balancing for observational
  studies.
\newblock {\em Biometrika}, 105(1):199--213.

\bibitem[\protect\astroncite{Yiu and Su}{2018}]{yiu2018covariate}
Yiu, S. and Su, L. (2018).
\newblock Covariate association eliminating weights: a unified weighting
  framework for causal effect estimation.
\newblock {\em Biometrika}.

\bibitem[\protect\astroncite{Zhao}{2018}]{zhao2018covariate}
Zhao, Q. (2018).
\newblock Covariate balancing propensity score by tailored loss functions.
\newblock {\em Annals of Statistics}, page in press.

\bibitem[\protect\astroncite{Zhao and Percival}{2017}]{zhao2017entropy}
Zhao, Q. and Percival, D. (2017).
\newblock Entropy balancing is doubly robust.
\newblock {\em Journal of Causal Inference}, 5(1).

\bibitem[\protect\astroncite{Zubizarreta}{2012}]{zubizarreta2012using}
Zubizarreta, J.~R. (2012).
\newblock Using mixed integer programming for matching in an observational
  study of kidney failure after surgery.
\newblock {\em Journal of the American Statistical Association},
  107(500):1360--1371.

\bibitem[\protect\astroncite{Zubizarreta}{2015}]{zubizarreta2015stable}
Zubizarreta, J.~R. (2015).
\newblock Stable weights that balance covariates for estimation with incomplete
  outcome data.
\newblock {\em Journal of the American Statistical Association},
  110(511):910--922.

\bibitem[\protect\astroncite{Zubizarreta
  et~al.}{2011}]{zubizarreta2011matching}
Zubizarreta, J.~R., Reinke, C.~E., Kelz, R.~R., Silber, J.~H., and Rosenbaum,
  P.~R. (2011).
\newblock Matching for several sparse nominal variables in a case-control study
  of readmission following surgery.
\newblock {\em The American Statistician}, 65(4):229--238.

\end{thebibliography}

\newpage

\appendix
\section*{Supplementary materials}
\section{Proof for the unconstrained dual formulation}
\label{sec:dualproof}

\subsection*{Proof of Theorem \ref{thm:shrink}}

\begin{proof}
We first present a vanilla form of the dual.
\begin{lemma}
\label{lemma:vanilladual}
The dual of the optimization problem (\ref{eq:covbalineq}) is 
\begin{equation*}   
\begin{aligned}
& \underset{\lambda}{\text{minimize}}
&& l(\lambda) \\
& \text{subject to}
&& \lambda \geq 0 \\
\end{aligned}
\end{equation*}
where 
\begin{equation*}
l(\lambda) = \frac{1}{n}\sum^n_{j=1} \{-Z_j n\rho(Q_j^\top \lambda) 
    + Q_j^\top \lambda\} + \lambda^\top d,
\end{equation*}
\begin{align*}
A_{K \times n} & =
    \left( \begin{array}{cccc}
    B_1(X_1) & B_1(X_2) & \ldots & B_1(X_n) \\
    \vdots & \vdots & \vdots & \vdots \\
    B_K(X_1) & B_K(X_2) & \ldots & B_K(X_n) \\ 
    \end{array} \right)_{K \times n},
\end{align*}
\begin{equation*}
Q_{2K \times n} =  \left( \begin{array}{r}
    A_{K \times n}\\
    -A_{K \times n}\\
    \end{array} \right)_{2K \times n},
\end{equation*}
and
\begin{equation*}
d_{2K \times 1} = \left( \begin{array}{r}
    \delta_{K \times 1}\\
    \delta_{K \times 1}\\
    \end{array} \right)_{2K \times 1}.
\end{equation*}
\end{lemma}
We prove this lemma towards the end of this section.

We then write $\lambda_{2K \times 1}  = \left( \begin{array}{r}
    \lambda_{+, K \times 1}\\
    \lambda_{-, K \times 1}\\
    \end{array} \right)_{2K \times 1}$. We have
\begin{align*}
l(\lambda) & = \frac{1}{n}\sum^n_{j=1} \{-Z_j n\rho(A_j^\top \lambda_+ 
    - A_j^\top \lambda_-) 
    + (A_j^\top \lambda_+ - A_j^\top \lambda_-) \} 
    + \lambda_+^\top \delta + \lambda_-^\top \delta\\
& = \frac{1}{n}\sum^n_{j=1} [ -Z_j n\rho \{A_j^\top (\lambda_+ -\lambda_-) \} 
    + A_j^\top (\lambda_+ - \lambda_-) ] 
    + (\lambda_+^\top + \lambda_-^\top) \delta.
\end{align*}
Suppose the optimizer is $\lambda^\dagger_{2K \times 1}  
    = \left( \begin{array}{r}
    \lambda^\dagger_{+, K \times 1}\\
    \lambda^\dagger_{-, K \times 1}\\
    \end{array} \right)_{2K \times 1}$. 
We claim that $\lambda^\dagger_{+, k} \cdot \lambda^\dagger_{-, k} = 0, 
k = 1, \ldots, K$, where the index $k$ points to the $k$th entry of a vector.

We prove this claim by contradiction. Suppose the opposite. If $\lambda^\dagger_{+, k} > 0 $ and 
$\lambda^\dagger_{-, k} > 0 $ for some $k$, then
$$
\lambda^{\dagger\dagger\top} = [\lambda^\dagger_+ 
    - \{0, \ldots, 0, \min(\lambda^\dagger_{+, k}, \lambda^\dagger_{-, k}), 0, \ldots, 0\},
    \lambda^\dagger_- 
    - \{ 0, \ldots, 0, \min(\lambda^\dagger_{+, k}, \lambda^\dagger_{-, k}), 0, \ldots, 0)\} ]
$$    
has
\[l(\lambda^{\dagger\dagger}) = l(\lambda^{\dagger}) - 2 \min(\lambda^\dagger_{+, k}, \lambda^\dagger_{-, k}) \cdot \delta
    < l(\lambda^{\dagger})\]
by $\delta > 0$ and $\min(\lambda^\dagger_{+, k}, \lambda^\dagger_{-, k}) > 0$.
This contradicts the fact that $\lambda^\dagger$ is the optimizer.
Theorem \ref{thm:shrink} then follows by rewriting $ \lambda_+ - \lambda_-$ as $\lambda$ and deducing $\lambda_+ + \lambda_- = |\lambda|$ from $\lambda^\dagger_{+, k} \cdot \lambda^\dagger_{-, k} = 0, 
k = 1, \ldots, K$.
\end{proof}

\subsection*{Proof of Lemma \ref{lemma:vanilladual}}

\begin{proof}

Rewriting problem (\ref{eq:covbalineq}) in matrix notation,
\begin{equation*}   
\begin{aligned}
& \underset{\boldsymbol{w}}{\text{minimize}}
&& \sum^n_{i=1} Z_i h(s_i) \\
& \text{subject to}
&& Q_{2K \times n} s_{n \times 1} \leq d_{2K \times 1} \\
\end{aligned}
\end{equation*}
where
\begin{align*}
s_{n \times 1} & = (s_i)_{n \times 1} 
    = (\frac{1}{n} - Z_i w_i)_{n \times 1},\\
A_{K \times n} & =
    \left( \begin{array}{cccc}
    B_1(X_1) & B_1(X_2) & \ldots & B_1(X_n) \\
    \vdots & \vdots & \vdots & \vdots \\
    B_K(X_1) & B_K(X_2) & \ldots & B_K(X_n) \\ 
    \end{array} \right)_{K \times n},
Q_{2K \times n} & =  \left( \begin{array}{r}
    A_{K \times n}\\
    -A_{K \times n}\\
    \end{array} \right)_{2K \times n},\\
d_{2K \times 1} & = \left( \begin{array}{r}
    \delta_{K \times 1}\\
    \delta_{K \times 1}\\
    \end{array} \right)_{2K \times 1}.\\
\end{align*}
Again as special cases, stable balancing weights have $h(x) =
(\frac{1}{n}-\frac{1}{r}-x)^2$ and entropy balancing has $h(x) =
(\frac{1}{n} - x)\log(\frac{1}{n} - x)$.

The problem is now in the form of \citet{tseng1987relaxation} and
\citet{tseng1991relaxation}.

The dual of this problem is 
\begin{equation*}   
\begin{aligned}
& \underset{\lambda}{\text{maximize}}
&& g(\lambda) \\
& \text{subject to}
&& \lambda \geq 0, \\
\end{aligned}
\end{equation*}
where $g(\lambda) = -\sum^n_{j=1} h^*_j(Q_j^\top \lambda) - <\lambda, d>,$
and $h_j^*(\cdot)$ is the convex conjugate of $Z_j h(\cdot).$
\begin{align*}
h^*_j(t) & = \sup_{s_j} \{t s_j - Z_j h(s_j)\}\\
& = \sup_{w_j} \{-t Z_j w_j + \frac{t}{n} - Z_j h(\frac{1}{n} - Z_j w_j)\}\\
& = \sup_{w_j} \{-t Z_j w_j + \frac{t}{n} - Z_j h(\frac{1}{n} - w_j)\}\\
& = -t Z_j w^*_j + \frac{t}{n} - Z_j h(\frac{1}{n} - w^*_j),
\end{align*}
where $w^*_j$ satisfies the first order condition
\begin{align*}
& -t Z_j + Z_j h'(\frac{1}{n} - w^*_j) = 0,\\
\Rightarrow & h'(\frac{1}{n} - w^*_j) = t,\\
\Rightarrow & w^*_j = \frac{1}{n} - (h')^{-1}(t).
\end{align*}
Therefore,
\begin{align*}
h^*_j(t) & = -t Z_j \frac{1}{n} + t Z_j (h')^{-1}(t) + \frac{t}{n} 
        - Z_j h \{ (h')^{-1}(t) \},\\
    & = -Z_j  [ \frac{t}{n} - t (h')^{-1}(t) + h\{(h')^{-1}(t) \} ]
        + \frac{t}{n}.
\end{align*}

Denote $\rho(\cdot)$ as 
\[\rho(t) = \frac{t}{n} - t (h')^{-1}(t) + h \{ (h')^{-1}(t) \}.\] 

This gives \[h^*_j(t) = -Z_j \rho(t) + \frac{t}{n}.\]

Also we notice that
\begin{align*}
\rho'(t) & = \frac{1}{n} - (h')^{-1}(t) - t \{ (h')^{-1}(t) \}' 
        + h'\{(h')^{-1}(t)\}\cdot \{(h')^{-1}(t)\}'\\
    & = \frac{1}{n} - (h')^{-1}(t) - t \{ (h')^{-1}(t) \}' + t \{ (h')^{-1}(t)\}'\\
    & = \frac{1}{n} - (h')^{-1}(t).\\
\end{align*}
This implies \[w^* = \rho'(t).\]

The dual formulation thus becomes
\begin{equation*}   
\begin{aligned}
& \underset{\lambda}{\text{minimize}}
&& l(\lambda) \\
& \text{subject to}
&& \lambda \geq 0 \\
\end{aligned}
\end{equation*}
where 
\begin{align*}
l(\lambda) = \frac{1}{n}\sum^n_{j=1} \{-Z_j n\rho(Q_j^\top \lambda) 
    + Q_j^\top \lambda\} + \lambda^\top d.
\end{align*}
\end{proof}

\section{Proof of the Asymptotic properties}
\label{sec:asymptotics_proof}

\subsection*{Proof of Theorem \ref{thm:wtconsistency}}

\begin{proof}
The proof utilizes the Bernstein's inequality as in
\citet{fan2016improving}.

We first prove the following lemma.
\begin{lemma}
\label{lemma:wtconsistlemma}
There exists a global minimizer $\lambda^\dagger$ such that
\[||\lambda^\dagger  - \lambda^*_1||_2 = O_p(K^{1/2}(\log K)/n + K^{1/2-r_\pi}).\]
\end{lemma}

\begin{proof}

Write $A_j = B(X_j) = \{ B_1(X_j), ..., B_K(X_j) \}$. Recall that the optimization objective is 
\[G(\lambda) := \frac{1}{n}\sum^n_{j=1} \{ -Z_j n\rho(A_j^\top \lambda) 
    + A_j^\top \lambda \} + |\lambda|^\top \delta,\]
    where $G(\cdot)$ is convex in $\lambda$ by the concavity of $\rho(\cdot).$
To show that a minimizer $\Delta^*$ of $G(\lambda^*_1 + \Delta)$ exists in $\mathcal{C} = \{\Delta\in\mathbb{R}^K:||\Delta||_2\leq  C K^{1/2}(\log K)/n + K^{1/2-r_\pi}\}$ for some constant $C$, it suffices to show that
\[E\{\inf_{\Delta\in\mathcal{C}}G(\lambda^*_1 + \Delta) - G(\lambda^*_1) > 0\}\rightarrow 1, \text{  as  } n\rightarrow\infty, (*)\]
by the continuity of $G(\cdot).$

To show $(*)$, we use mean value theorem: for some $\tilde{\lambda}$ between $\lambda^\dagger$ and $\lambda_1^*$,
\begin{align*}
&G(\lambda^*_1 + \Delta) - G(\lambda^*_1)\\
\geq & \Delta\cdot\frac{1}{n}\sum^n_{j=1} \{-Z_j n\rho'(A_j^\top \lambda^*_1) A_j
    + A_j\} + \frac{1}{2}\Delta^\top \cdot\{\sum^n_{j=1} -Z_j \rho''(A_j^\top \tilde{\lambda}) A_j^\top A_j\}\cdot\Delta - |\Delta|^\top \delta\\
\geq& -||\Delta||_2 \cdot ||\frac{1}{n}\sum^n_{j=1} -Z_j n\rho'(A_j^\top \lambda^*_1) A_j
    + A_j||_2 \nonumber\\
    &+ \frac{1}{2}\Delta^\top \cdot\{\sum^n_{j=1} -Z_j \rho''(A_j^\top \tilde{\lambda}) A_j^\top A_j\}\cdot\Delta - ||\Delta||_2||\delta||_2\\
\geq & -||\Delta||_2 \cdot ||\frac{1}{n}\sum^n_{j=1} -Z_j n\rho'(A_j^\top \lambda^*_1) A_j
    + A_j||_2 - ||\Delta||_2||\delta||_2.
\end{align*}

The first inequality is due to the triangle inequality,
$|\lambda_1^*+\Delta| - |\lambda_1^*|\geq -|\Delta|$. The second
inequality follows from Cauchy-Schwarz inequality. The third
inequality is due to the positivity of $\frac{1}{2}\Delta^\top \cdot
\{ \sum^n_{j=1} -Z_j \rho''(A_j^\top \tilde{\lambda}) A_j^\top A_j \}
\cdot\Delta$ by Assumption \ref{assumption:wtconsistency}.3.

Next we notice that
\begin{align*}
& ||\frac{1}{n}\sum^n_{j=1} \{-Z_j n\rho'(A_j^\top \lambda^*_1) A_j + A_j \} ||_2\\
\leq & ||\frac{1}{n}\sum^n_{j=1} ( -Z_j \frac{1}{\pi_j} A_j + A_j ) ||_2 + ||\frac{1}{n}\sum^n_{j=1} -Z_j \{ \frac{1}{\pi_j} - n\rho'(A_j^\top \lambda^*_1) \} A_j||_2\\
\leq & ||\frac{1}{n}\sum^n_{j=1} (1-\frac{Z_j}{\pi_j}) A_j||_2 + \frac{1}{n}\sum_{j=1}^n||A_j||_2O(K^{-r_\pi}).
\end{align*}
The first inequality is due to the triangle inequality. The second inequality is due to Assumption \ref{assumption:wtconsistency}.3 and \ref{assumption:wtconsistency}.6.

We first use the Bernstein's inequality to bound both terms.

Recall that the Bernstein's inequality for random matrices in \citet{tropp2015introduction} says the following. Let $\{Z_k\}$ be a sequence of independent random matrices with dimensions $d_1\times d_2$. Assume that $EZ_k = 0$ and $||Z_k||_2\leq R_n$ almost surely. Define 
\[\sigma_n^2 = \max\{||\sum^n_{k=1}E(Z_kZ_k^\top)||_2, ||\sum^n_{k=1}E(Z_k^\top Z_k)||_2\}.\]
Then for all $t\geq 0$,
\[pr(||\sum^n_{k=1}Z_k||_2\geq t)\leq (d_1+d_2)\exp(-\frac{t^2/2}{\sigma^2_n+R_nt/3}).\]

For the first term $||\frac{1}{n}\sum^n_{j=1} (1-\frac{Z_j}{\pi_j}) A_j||_2$, we notice that 
\begin{align}
\label{eq:bern11}
E\{\frac{1}{n}(1-\frac{Z_j}{\pi_j}) A_j\} = E[E\{\frac{1}{n}(1-\frac{Z_j}{\pi_j}) A_j\mid X_j\}] = 0.
\end{align}
The last equality is because $E(Z_j) =\pi_j. $

Then for $||\frac{1}{n}\sum^n_{j=1} (1-Z_j/\pi_j) A_j||_2$, we have
\begin{align}
&||\frac{1}{n} (1-\frac{Z_j}{\pi_j}) A_j||_2\nonumber\\
\leq &\frac{1}{n}|| (1-\frac{Z_j}{\pi_j})||_2||A_j||_2\nonumber\\
\leq &\frac{1}{n} (\frac{1 - \pi_j}{\pi_j}) CK^{1/2}\nonumber\\
= &\frac{1}{n} \{ n\rho'(A_j^\top\lambda_1^*) - 1 \} CK^{1/2}\nonumber\\
\leq & C'\frac{K^{1/2}}{n}.
\label{eq:bern12}
\end{align}

The first inequality is due to Cauchy-Schwarz inequality. The second inequality is due to Assumption \ref{assumption:wtconsistency}.4 and $E(1-Z_j/\pi_j)^2 = \text{var}(1-Z_j/\pi_j) = \pi_j(1-\pi_j)/\pi_j^2 =(1 - \pi_j)/\pi_j$. The third equality is due to $\pi_j = \{n\rho'(A_j^\top\lambda_1^*)\}^{-1}.$ The fourth inequality is due to Assumption \ref{assumption:wtconsistency}.3.

Finally, for $||\sum^n_{k=1}E \{ \frac{1}{n^2}(1-\frac{Z_j}{\pi_j} )^2A_jA_j^\top\}||_2$, we have
\begin{align}
&||\sum^n_{k=1}E\{\frac{1}{n^2}(1-\frac{Z_j}{\pi_j})^2A_jA_j^\top\}||_2\nonumber\\
\leq & \frac{1}{n}\sup_j(1-\frac{Z_j}{\pi_j})^2||E(A_jA_j^\top)||_2\nonumber\\
\leq & \frac{C''}{n}.
\label{eq:bern13}
\end{align}
The first inequality is taking the sup over $(1-\frac{Z_j}{\pi_j})^2$. The second inequality is due to Assumption \ref{assumption:wtconsistency}.3, \ref{assumption:wtconsistency}.4, and $\pi_j = \{n\rho'(A_j^\top\lambda_1^*)\}^{-1}.$

\Cref{eq:bern11}, \Cref{eq:bern12}, and \Cref{eq:bern13}, together with the Bernstein's inequality, imply
\begin{align*}
\textrm{pr} \{ ||\frac{1}{n}\sum^n_{j=1} (1-\frac{Z_j}{\pi_j}) A_j||_2\geq t \} \leq (K+1)\exp(-\frac{t^2/2}{\frac{C''}{n} + C'\frac{K^{1/2}}{n}\cdot t/3}).
\end{align*}
The right side goes to zero as $K\rightarrow\infty$ when \[\frac{t^2/2}{\frac{C''}{n} + C'\frac{K^{1/2}}{n}\cdot t/3}\geq\log K.\] It suffices when $t = O_p \{ K^{1/2}(\log K)/n \}.$ 

Therefore, we have
\begin{align}
\label{eq:bern1}
||\frac{1}{n}\sum^n_{j=1} (1-\frac{Z_j}{\pi_j}) A_j||_2 = O_p\{K^{1/2}(\log K)/n\}.
\end{align}

Now we work on the second term $\frac{1}{n}\sum_{j=1}^n||A_j||_2O(K^{-r_\pi})$. We have
\begin{align}
\label{eq:bern2}
\frac{1}{n}\sum_{j=1}^n||A_j||_2O(K^{-r_\pi})\leq CK^{1/2-r_\pi}.
\end{align}
This inequality is due to Assumption \ref{assumption:wtconsistency}.4.

Combining \Cref{eq:bern1}, \Cref{eq:bern2}, and Assumption \ref{assumption:wtconsistency}.7, we have 
\begin{align*}
&G(\lambda^*_1 + \Delta) - G(\lambda^*_1)\\
=&-||\Delta||_2 \cdot O_p ( \frac{K^{1/2}\log K}{n} + K^{1/2-r_\pi}) + \frac{1}{2}||\Delta||^2_2||\delta||_2\\
\geq & 0
\end{align*}
for $\Delta = C\frac{K^{1/2}\log K}{n} + K^{1/2-r_\pi}$ with large enough constant $C > 0.$ 

$(*)$ is thus proved.
\end{proof}

Now we prove \Cref{thm:wtconsistency}.

\begin{align*}
&\sup_{x\in\mathcal{X}}|nw^*(x) - \frac{1}{\pi(x)}|\\
= & \sup_{x\in\mathcal{X}} |n\rho'\{ B(x)^\top\lambda^\dagger \}  - n\rho'\{ m^*(x) \}|\\
\leq & \sup_{x\in\mathcal{X}} |n\rho'\{ B(x)^\top\lambda^\dagger \}  - n\rho'\{B(x)^\top\lambda^*_1\}| + \sup_{x\in\mathcal{X}} |n\rho'\{B(x)^\top\lambda^*_1\}  - n\rho'\{ m^*(x) \}|\\
= & O\{ \sup_{x\in\mathcal{X}} |B(x)^\top\lambda^\dagger  - B(x)^\top\lambda^*_1| \} + O(K^{-r_\pi}) \\
\leq & O \{ \sup_{x\in\mathcal{X}} ||B(x)||_2||\lambda^\dagger  - \lambda^*_1||_2 \} + O(K^{-r_\pi}) \\
= & O_p \{ K(\frac{\log K}{n} + K^{-r_\pi}) \} + O(K^{-r_\pi}) \\
= & O_p ( \frac{K\log K}{n} + K^{1-r_\pi})  \\
= & o_p(1)
\end{align*}

The first equality rewrites $\pi(x) =
\{n\rho'(B(x)^\top\lambda_1^*)\}^{-1}.$ The second inequality is due
to the triangle inequality. The third inequality is due to Assumptions
\ref{assumption:wtconsistency}.3 and \ref{assumption:wtconsistency}.6.
The fourth inequality is due to the Cauchy-Schwarz inequality. The
fifth equality is due to \Cref{lemma:wtconsistlemma} and Assumption
\ref{assumption:wtconsistency}.4. The sixth equality holds because the
first term dominates the second. The seventh equality is due to
Assumptions \ref{assumption:wtconsistency}.5 and
\ref{assumption:wtconsistency}.6.

Also, we have

\begin{align*}
&||nw^*(x)-\frac{1}{\pi(x)}||_{P,2}\\
=&||n\rho'\{\lambda^{\dagger\top}B(X)\}-\frac{1}{\pi(x)}||_{P,2}\\
\lesssim & ||n\rho'\{\lambda^{\dagger\top}B(X)\}-n\rho' \{ \lambda_1^{*\top}B(X) \}||_{P,2} + ||\frac{1}{\pi(x)}-n\rho' \{ \lambda_1^{*\top}B(X) \}||_{P,2}\\
\lesssim &||(\lambda^{\dagger}-\lambda_1^*)^\top B(X)||_{P,2}+ \sup_{x\in\mathcal{X}}|m^*(x)-\lambda_1^{*\top}B(x)|\\
= & O_p \{ K^{1/2}(\frac{\log K}{n} + K^{-r_\pi}) \} + O(K^{-r_\pi}) \\
= & O_p(\frac{K^{1/2}\log K}{n} + K^{1/2-r_\pi})\\
= & o_p(1)
\end{align*}

The first equality rewrites $\pi(x) =
[n\rho'\{B(x)^\top\lambda_1^*\}]^{-1}.$ The second inequality is due
to the triangle inequality. The third inequality is due to Assumption
\ref{assumption:wtconsistency}.3. The fourth inequality is due to
\Cref{lemma:wtconsistlemma}, Assumption
\ref{assumption:wtconsistency}.4 and Assumption
\ref{assumption:wtconsistency}.6. The fifth equality is due to the
first term dominates the second. The sixth equality is due to
Assumption \ref{assumption:wtconsistency}.5 and Assumption
\ref{assumption:wtconsistency}.6.

\end{proof}

\subsection*{Proof of Theorem \ref{thm:asymptotics}}

\begin{proof}

The proof utilizes empirical processes techniques as in
\citet{fan2016improving}.

We first decompose $\hat{Y}_{w^*} - \bar{Y}$ into several residual terms.
\begin{align*}
\hat{Y}_{w^*} - \bar{Y} & = \sum^n_{i=1}  Z_i w^*_i Y_i - \bar{Y}\\
& = \sum^n_{i=1} Z_i w^*_i \{Y_i - Y(X_i)\} 
    + \sum^n_{i=1} (Z_i w^*_i - \frac{1}{n} ) Y(X_i)
    + \{\frac{1}{n} \sum^n_{i=1} Y(X_i) - \bar{Y}\}\\
\begin{split} & = \sum^n_{i=1} Z_i w^*_i \{Y_i - Y(X_i)\} 
    + \sum^n_{i=1} (Z_i w^*_i - \frac{1}{n} ) 
        \{ Y(X_i) - \lambda_2^{*\top} B(X_i) \}\\
    & + \sum^n_{i=1} (Z_i w^*_i - \frac{1}{n} ) \lambda_2^{*\top} B(X_i)
      +  \{ \frac{1}{n} \sum^n_{i=1} Y(X_i) - \bar{Y} \} \end{split}\\
& = \frac{1}{n} \sum^n_{i=1}S_i + R_0 + R_1 + R_2,
\end{align*}
where
\begin{align*}
S_i & = \frac{Z_i}{\pi_i} \{Y_i - Y(X_i)\} + \{Y(X_i) - \bar{Y}\},\\
R_0 & = \sum^n_{i=1} (w^*_i - \frac{1}{n\pi_i}) Z_i
        \{Y_i - Y(X_i)\},\\
R_1 & = \sum^n_{i=1} (Z_i w^*_i - \frac{1}{n})
        \{ Y(X_i) - \lambda_2^{*\top} B(X_i) \},\\
R_2 & = \sum^n_{i=1} (Z_i w^*_i - \frac{1}{n} )
        \{ \lambda_2^{*\top} B(X_i) \}.\\
\end{align*}
Below we show $R_j = o_p(n^{-1/2}), 0 \leq j \leq 2$. The conclusion
follows from $S_i$ taking the same form as the efficient score
\citep{hahn1998role}. $\hat{Y}_{w^*}$ is thus asymptotically normal
and semiparametrically efficient.

We first study
$R_0 =\sum^n_{i=1} (nw^*_i - 1/\pi_i) Z_i
        \{Y_i - Y(X_i)\}/n$.
Consider an empirical process 
$\mathbb{G}_n(f_0) = n^{1/2} [\sum^n_{i=1} f_0(Z_i, Y_i, X_i)/n -E\{ f_0(Z, Y, X)\}]$,
where 
\[f_0(Z, Y, X) = Z \{Y - Y(X)\}\left[n\rho'\{m(X)\} - \frac{1}{\pi(x)}\right].\]

By the missing at random assumption, we have $Ef_0\{Z, Y,X\} = 0.$

By \Cref{thm:wtconsistency}, we have 
\[\sup_{x\in\mathcal{X}}|\rho'\{B(x)^\top\lambda^{\dagger}\} - \frac{1}{n\pi(x)}| =  O_p(\frac{K\log K}{n} + K^{1-r_\pi}) = o_p(1).\]

By Markov's inequality and maximal inequality, we have
\[n^{1/2}R_0\leq\sup_{f_0\in\mathcal{F}}\mathbb{G}_n(f_0)\lesssim E\sup_{f_0\in\mathcal{F}}\mathbb{G}_n(f_0)\lesssim J_{[]}\{||F_0||_{P,2},\mathcal{F},L_2(P)\},\]
where the set of functions is $\mathcal{F}=\{f_0:||m-m^*||_\infty\leq\delta_0\}$, where $||f||_{\infty}=\sup_{x\in\mathcal{X}}|f(x)|$ and $\delta_0=C\{K(\log K)/n + K^{1-r_\pi}\}$ for some constant $C>0$.

The second inequality is due to Markov's inequality.
$J_{[]}\{||F_0||_{P,2},\mathcal{F},L_2(P)\}$ is the bracketing
integral. $F_0:= \delta_0|Y-Y(X)|\gtrsim |f_0(Z, Y, X)|$ is the
envelop function. We also have
$||F_0||_{P,2}=(EF_0^2)^{1/2}\lesssim\delta_0$ by $E|Y-Y(X)|<\infty.$

Next we bound $J_{[]}\{||F_0||_{P,2},\mathcal{F},L_2(P)\}$ by $n_{[]}\{\varepsilon,\mathcal{F}, L_2(P)\}$:
\[J_{[]}\{||F_0||_{P,2},\mathcal{F},L_2(P)\} \lesssim \int^{\delta}_0 [ n_{[]}\{\varepsilon,\mathcal{F}, L_2(P)\} ]^{1/2}\text{d}\varepsilon.\]

Define a new set of functions
$\mathcal{F}_0=\{f_0:||m-m^*||_\infty\leq C\}$ for some constant
$C>0$. Then,
\begin{align*}
\log n_{[]}\{\varepsilon,\mathcal{F}, L_2(P)\}&\lesssim \log n_{[]} \{ \varepsilon,\mathcal{F}_0\delta_0, L_2(P) \} \\
&=\log n_{[]}\{ \varepsilon/\delta_0,\mathcal{F}_0, L_2(P) \}\\
&\lesssim \log n_{[]}\{\varepsilon/\delta_0,\mathcal{M}, L_2(P) \}\\
&\lesssim(\delta_0/\varepsilon)^{(1/k_1)}.
\end{align*}
The first inequality is due to the fact that $\rho'(\cdot)$ bounded
away from 0 and Lipschitz. The last inequality is due to Assumption
\ref{assumption:regularity}.2.

Therefore, we have
\[J_{[]}\{||F_0||_{P,2},\mathcal{F},L_2(P)\}\lesssim\int^{\delta}_0 [ \log n_{[]}\{\varepsilon,\mathcal{F}, L_2(P)\} ]^{1/2}\text{d}\varepsilon\lesssim\int^{\delta}_0(\delta_0/\varepsilon)^{(1/2k_1)}\text{d}\varepsilon.\]

This goes to 0 as $\delta$ goes to 0 by $2k_1>1$ and the integral converges. Thus, this shows that $n^{1/2}R_0=o_p(1)$.

Next, we consider $R_1=\sum^n_{i=1} (nZ_i w^*_i - 1)
        \{Y(X_i) - \lambda_2^{*\top} B(X_i)\}/n$. 
Define the empirical process $\mathbb{G}_n(f_1)=n^{1/2}[\sum^n_{i=1}f_1(Z_i,X_i)/n-E\{f_1(Z,X)\}]$, where $f_1(Z,X)=[nZ\rho'\{m(x)\}-1]\{Y(X)-\lambda_2^{*\top} B(X)\}$. 

Write $\Delta(X) := Y(X)-\lambda_2^{*\top} B(X).$ By Assumption \ref{assumption:regularity}.3, we have $||\Delta||_\infty\lesssim K^{-r_y}$. 

By \Cref{thm:wtconsistency}, we have
\[\left\|n\rho'\{\lambda^{\dagger\top}B(X)\}-\frac{1}{\pi(x)}\right\|_{P,2} = O_p(\frac{K\log K}{n} + K^{1-r_\pi}).\]

Therefore, we have
\begin{align*}
n^{1/2}R_1 &= \mathbb{G}_n(f_1)+n^{1/2} Ef_1(Z,X)\\
&\leq\sup_{f_1\in\mathcal{F}_1}\mathbb{G}_n(f_1)+n^{1/2}\sup_{f_1\in\mathcal{F}_1}Ef_1,
\end{align*}
where $\mathcal{F}_1 = \{f_1:||m-m^*||_{P,2}\leq \delta_1,$$ ||\Delta||_\infty\leq\delta_2\}, $$\delta_1 = C \{K^{1/2}(\log K)/n + K^{1/2-r_\pi}\}$, $\delta_2 = CK^{-r_y}$ for some constant $C > 0.$

Again, by Markov's inequality and the maximal inequality,
\begin{align*}
\sup_{f_1\in\mathcal{F}_1}\mathbb{G}_n(f_1)\lesssim E\sup_{f_1\in\mathcal{F}_1}\mathbb{G}_n(f_1)\lesssim J_{[]}\{||F_1||_{P,2},\mathcal{F},L_2(P)\},
\end{align*}
where $F_1:=C\delta_2$ for some constant $C > 0$ so that $||F_1||_{P,2} \lesssim \delta_2.$

Similar to characterizing $R_1$, we we bound $J_{[]}(||F_1||_{P,2},\mathcal{F}_1,L_2(P)$ by $n_{[]}(\varepsilon,\mathcal{F}_1, L_2(P))$:
\[J_{[]}(||F_1||_{P,2},\mathcal{F}_1,L_2(P) \lesssim \int^\delta_0 \{n_{[]}(\varepsilon,\mathcal{F}_1, L_2(P))\}^{1/2}\text{d}\varepsilon.\]

Then, we bound $n_{[]}(\varepsilon,\mathcal{F}_1, L_2(P))$:
\begin{align*}
\log n_{[]}\{\varepsilon, \mathcal{F}_1, L_2(P)\}&\lesssim \log n_{[]}\{\varepsilon/\delta_2, \mathcal{F}_0, L_2(P)\}\\
&\lesssim \log n_{[]}\{\varepsilon/\delta_2, G_{10}, L_2(P)\} + \log n_{[]}\{\varepsilon/\delta_2, G_{20}, L_2(P)\}\\
&\lesssim \log n_{[]}\{\varepsilon/\delta_2, \mathcal{M}, L_2(P)\} + \log n_{[]}\{\varepsilon/\delta_2, \mathcal{H}, L_2(P)\}\\
&\lesssim (\delta_1/\varepsilon)^{1/k_1}+(\delta_2/\varepsilon)^{1/k_2}.
\end{align*}

where 
\[\mathcal{F}_0=\{f_1:||m-m^*||_{P,2}\leq C, ||\Delta||_{P,2}\leq 1\},\]
\[\mathcal{G}_{10}=\{m\in \mathcal{M}+m^*:||m||_{P,2}\leq C\},\]
\[\mathcal{G}_{20}=\{\Delta\in\mathcal{H}-\lambda_2^{*\top} B(x):||\Delta||_{P,2}\leq 1\}.\]

The second inequality is due to $\rho'$ is Lipschitz and bounded away from 0.

Therefore we have
\[J_{[]}\{||F_1||_{P,2},\mathcal{F}_1,L_2(P)\}\lesssim\int^\delta_0(\delta_1/\varepsilon)^{(1/2k_1)}d\varepsilon + \int^\delta_0(\delta_2/\varepsilon)^{(1/2k_2)}d\varepsilon.\]

By $2k_1>1$ and $2k_2>1$, we have $J_{[]}\{||f_1||_{P,2}, \mathcal{F}, L_2(P)\} = o(1)$. This gives $\sup_{f_1\in\mathcal{F}_1}\mathbb{G}_n(f_1) = o_p(1)$.

Now we look at $n^{1/2}\sup_{f_1\in\mathcal{F}_1}Ef_1,$
\begin{align*}
n^{1/2}\sup_{f_1\in\mathcal{F}} Ef_1 
&=n^{1/2}\sup_{m\in\mathcal{G}_1, \Delta\in\mathcal{G}_2}E\{\pi(X)[n\rho'\{m(X)\}-1]\Delta(X)\}\\
&=n^{1/2}\sup_{m\in\mathcal{G}_1, \Delta\in\mathcal{G}_2}E([n\rho'\{m(x)\}-\frac{1}{\pi(x)}]\pi(x)\Delta(x))\\
&\lesssim n^{1/2} \sup_{m\in\mathcal{G}_1}||n\rho'\{m(x)\}-\frac{1}{\pi(x)}||_{P,2} \sup_{\Delta\in\mathcal{G}_2}||\Delta(x)||_{P,2}\\
&\lesssim n^{1/2}\delta_1\delta_2 = o_p(1),
\end{align*}
where 
$\mathcal{G}_1=\{m\in\mathcal{M}: ||m-m^*||_{P,2}\leq\delta_1\}, \mathcal{G}_2=\{\Delta\in\mathcal{H}-\lambda_2^{*\top} B(x):||\Delta||_\infty\leq\delta_2\}$.

The last equality is due to the assumption $n^{1/2}\lesssim
K^{r_\pi+r_y-1/2}$.

Therefore, we can conclude that $n^{1/2}R_1=o_p(1).$

Lastly, $R_2 = \lambda_2^{*\top}\{\sum^n_{i=1} (Z_i w^*_i - 1/n )
B(X_i)\} = o_p(1)$ by $\sum^n_{i=1} (Z_i w^*_i - 1/n ) B(X_i)\leq
||\delta||^2 = o_p(1)$ due to the constraints posited in the
optimization problem for minimal weights.

We finally prove the consistency of the variance estimator. We need a
stronger smoothness assumption, i.e. $r_y>1$.

Under assumptions \ref{assumption:wtconsistency} and
\ref{assumption:regularity}, we construct a variance estimator based
on a direct approximation of the efficient influence function. Recall
that the efficient influence function determines the semiparametric
efficiency bound \citep{hahn1998role}:
\begin{align*}
V_{opt} & := var(Y(X_i)) + E\{var(Y| X_i)/\pi(X_i)\}\\
& = E\left\{\left(\frac{Z_iY_i}{\pi(X_i)} - \bar{Y} - Y(X_i)(\frac{Z_i}{\pi(X_i)}-1)\right)^2\right\}.
\end{align*}

We estimate $V_{opt}$ with $\hat{V}_K$:
\begin{align*}
\begin{split}
\hat{V}_K = &\frac{1}{n}\sum^n_{i=1} \left[ nZ_iw_iY_i - \sum_{i=1}^nw_iY_i \right.\\
&\left. - B(X_i)^\top \left\{ \frac{1}{n}\sum_{i=1}^nZ_iw_iB(X_i)^\top B(X_i) \right\}^{-1}
\right.
\left.
\left\{ \frac{1}{n}\sum_{i=1}^nZ_iw_iB(X_i)^\top Y_i \right\}
(nZ_iw_i - 1) \right]^2.
\end{split}
\end{align*}
In particular, $\{\frac{1}{n}\sum_{i=1}^nZ_iw_iB(X_i)^\top
B(X_i)\}^{-1}\{\frac{1}{n}\sum_{i=1}^nZ_iw_iB(X_i)^\top Y_i\}$ is a least
square estimator of $Y(X_i)$.

To show $\hat{V}_K$ is consistent with $V_{opt}$, it is sufficient to show 
\begin{align*}
|B(X_i)^\top\{\frac{1}{n}\sum_{i=1}^nZ_iw_iB(X_i)^\top
B(X_i)\}^{-1}\{\frac{1}{n}\sum_{i=1}^nZ_iw_iB(X_i)^\top
Y_i\}-Y(X_i)|\stackrel{a.s.}{\rightarrow} 0.\qquad (**)
\end{align*}
This is because  $nw_i$ is a consistent estimator of $1/\pi(X_i)$ by
\Cref{thm:wtconsistency} and $\sum_{i=1}^nw_iY_i$ is a consistent
estimator of $\bar{Y}$ by \Cref{thm:asymptotics}.

Below we prove $(**)$. 

We first rewrite $Y_i$ as 
$Y_i = B(X_i)^\top\lambda_2^* + \gamma + \epsilon_i,$
where $\gamma = O(K^{-r_y})$ from Assumption
\ref{assumption:regularity}.3, and $\epsilon_i$ is some iid zero mean
error with variance $\sigma^2 = var(Y|X_i)$. Therefore,
\begin{align*}
&\{\frac{1}{n}\sum_{i=1}^nZ_iw_iB(X_i)^\top
B(X_i)\}^{-1}\{\frac{1}{n}\sum_{i=1}^nZ_iw_iB(X_i)^\top
Y_i\}\\
=&\{\frac{1}{n}\sum_{i=1}^nZ_iw_iB(X_i)^\top
B(X_i)\}^{-1}[\frac{1}{n}\sum_{i=1}^nZ_iw_iB(X_i)^\top
\{B(X_i)^\top\lambda_2^* + \gamma + \epsilon_i\}]\\
=&\lambda^*_2 + \{\frac{1}{n}\sum_{i=1}^nZ_iw_iB(X_i)^\top
B(X_i)\}^{-1}\{\frac{1}{n}\sum_{i=1}^nZ_iw_iB(X_i)^\top\}\gamma
\\
&+ \{\frac{1}{n}\sum_{i=1}^nZ_iw_iB(X_i)^\top
B(X_i)\}^{-1}\{\frac{1}{n}\sum_{i=1}^nZ_iw_iB(X_i)^\top\epsilon_i\}\\
=&\lambda^*_2 + E\{Z_iw_iB(X_i)^\top
B(X_i)\}^{-1}E\{Z_iw_iB(X_i)^\top\}\{\gamma+E(\epsilon_i)\} + O_p(n^{-1/2})\\
=&\lambda^*_2 + O_p(K^{-r_{y}+1/2}).
\end{align*}

The last equality is due to assumptions
\ref{assumption:wtconsistency}.4 and \ref{assumption:regularity}.3 and
the law of large numbers.

Finally we have 
\begin{align*}
&B(X_i)^\top\{\frac{1}{n}\sum_{i=1}^nZ_iw_iB(X_i)^\top
B(X_i)\}^{-1}\{\frac{1}{n}\sum_{i=1}^nZ_iw_iB(X_i)^\top
Y_i\}\\
=&B(X_i)^\top\lambda^*_2 + B(X_i) \cdot O_p(K^{-r_{y}+1/2})\\
=&Y(X_i) + B(X_i) \cdot O_p(K^{-r_{y}+1/2}) + O_p(K^{-r_{y}})\\
=&Y(X_i) + o_p(1)
\end{align*}
The last equality is due to assumption
\ref{assumption:wtconsistency}.4 and the additional assumption $r_y>1$.
\end{proof}

\section{\Cref{thm:oracle_words} explained}

\label{sec:oracle}

Due to the connection to shrinkage estimation, for each basis function
that we balance, we implicitly include a corresponding term in the
inverse propensity score model. In practice, we are often concerned
about the estimation loss due to fitting an overly complex model. In
the context of minimal weights, an overly complex model corresponds to
balancing more terms than are needed.

\Cref{thm:oracle_words} is an oracle inequality that bounds this loss
and states that approximate balancing---as opposed to exact
balancing---mimics the act of upper bounding the number of effective
balancing constraints. Hence, minimal weights do not suffer much from
excessive balancing when few constraints are active. We also remark
that this sparsity assumption on the balancing constraints is commonly
satisfied in real data sets. This is exemplified by the sparsity of
the shadow prices in the 2010 Chilean post earthquake survey data; see
Figure 1 of \citet{zubizarreta2015stable}.

The oracle inequality we proved in \Cref{sec:practice} leverages an
oracle inequality for lasso in the high dimensional generalized linear
model literature \citep{van2008high}. The original oracle inequality
says the lasso estimator (with $\ell^1$ penalty) under general
Lipschitz losses behaves similarly to the estimator with $\ell^0$
penalty, if the true generalized linear model is sparse.

Recall that the minimal weights compute
$$\lambda^\dagger := \argmin G(\lambda) =  \argmin \sum^n_{j=1} \left\{-Z_j \rho(A_j^\top \lambda)
    + A_j^\top \lambda \cdot \frac{1}{n} \right\} + |\lambda|^\top \delta.$$
This is a lasso estimator under the loss function 
\[L_{w}(x, z) = -z\cdot n(\rho\circ(\rho')^{-1}\circ w)(x) + ((\rho')^{-1}\circ w)(x),\]
where the fit for $w$ is $\hat{w}(x) = \rho'(B(x)^\top\hat{\lambda})$. This loss function is the same loss function as in \Cref{eq:dualloss} but written as a function of $w$.
Correspondingly, the empirical loss is
\[\sum_n^{i=1}L_w(X_i, Z_i) = \frac{1}{n}\sum^n_{i=1} \left\{-Z_i n\cdot (\rho\circ(\rho')^{-1}\circ w)(X_i) + ((\rho')^{-1}\circ w)(X_i) \right\},\]
and the theoretical risk is
\begin{align*}
EL_w(X,Z) = &\frac{1}{n}\sum^n_{i=1} E\left\{-Z_in \cdot (\rho\circ(\rho')^{-1}\circ w)(X_i) + ((\rho')^{-1}\circ w)(X_i)\right\}\\
=&\frac{1}{n}\sum^n_{i=1}\left\{-\pi(X_i)\cdot n(\rho\circ(\rho')^{-1}\circ w)(X_i) + ((\rho')^{-1}\circ w)(X_i)\right\}.
\end{align*}
We define the target $w^0$ as the minimizer of the theoretical risk
\[w^0(x):=\argmin EL_w(X,Z) = \frac{1}{n\pi(x)}.\]
The last equality is due to setting $\partial EL_w(X,Z)/\partial w =
0.$ This is the true inverse propensity score function used for
inverse probability weights. We are interested in studying the excess
risk of estimators
\[\mathcal{E}(w) := E\{L_w(X,Z) - L_{w^0}(X,Z)\}.\]

For simplicity of notation, we write $w_\lambda(x) = \rho'(B(x)^\top\lambda), \lambda\in\mathbb{R}^K$. Approximate balancing weights thus perform the empirical risk minimization of
\[\lambda^\dagger := \argmin_{\lambda}\{\frac{1}{n}\sum^n_{i=1}L_{w_\lambda}(X_i,Z_i)+|\lambda|^\top \delta\}.\]
 We look at the case of $\delta = \delta^+(\hat{\sigma}_1, \hat{\sigma}_2, ..., \hat{\sigma}_K),$ for some $ \delta^+ > 0$, where $\hat{\sigma}_k$ is the (sample) standard error of $B_k(X), k = 1, ..., K$. This aligns closely with the common way of setting $\delta$; we specify approximate balancing constraints in units of the standard error of each covariate.

We consider the following oracle estimator
\[\lambda^* := \argmin_\lambda\{EL_{w_\lambda}(X,Z)+||\lambda||_0 \cdot C_0\},\]
for some constant $C_0 > 0.$ $\lambda^*$ can also be seen as the minimizer of $\mathbb{P}L_{w_\lambda}$ under the constraint that $||\lambda||_0\leq C_1$ for some $C_1 > 0$. 

$||\lambda||_0$ is the number of nonzero entries of $\lambda.$ This is
also the number of active or effective covariate balancing constraints
in the optimization problem (\ref{eq:covbalineq}). In this sense, the
oracle estimator roughly performs the same covariate balancing exactly
as its approximate counterpart $\lambda^\dagger$ but the number of
effective constraints being capped by some constant.

We now assume the following conditions hold and present the oracle
inequality.

\begin{assumption}
\label{assumption:oracleassumption}
The following conditions hold:
\begin{enumerate}
\item There exist constants $0 < c_0 < 1/2,$ such that $c_0\leq
n\rho'(v) \leq 1 - c_0$ for any $v = B(x)^\top\lambda$ with
$\lambda\in int(\Theta)$. Also, there exist constants $c_1 < c_2 < 0$,
such that $c_1 \leq n\rho''(v) \leq c_2<0$ in some small neighborhood
$\mathcal{B}$ of $v^* = B(x)^\top\lambda^\dagger$.
   \item $\epsilon_0 < \pi(x) < 1 - \epsilon_0, \forall x$, for some constant $0 < \epsilon_0 < 1$,
   \item $M := \max ||B_k(x)||_\infty/\sigma_k <\infty,$ where $\sigma_k$ is the (population) standard deviation of $B_k(X), k = 1, ..., K.$
\end{enumerate}
\end{assumption}

Commenting on the previous assumptions, Assumption 3.1 is similar to
Assumption 1.3.  Assumption 3.2 is similar to the overlap condition of
propensity scores. Both of them ensure the quadratic margin condition
required by the lasso oracle inequality (the quadratic margin
condition says in the $\ell^\infty$ neighborhood of $w^0$ the excess
risk $\mathcal{E}$ is bounded from below by a quadratic function).
Assumption 3.3 is similar to Assumption 1.4. It ensures the existence
of the constant $\bar{\lambda} > 0$ in the theorem.

We further assume the following technical conditions.
\begin{assumption} Assume the following technical conditions hold.
\label{assumption:oracletechnical}
\begin{enumerate}
\item There exists $\eta>0$ such that $n||w_{\lambda^*} - w^0||_\infty \leq \eta$ and $n||w_{\tilde{\lambda} - w^0}||_\infty \leq \eta$, where $\tilde{\lambda} = \argmin_{\lambda\in\Theta: \sum_k \sigma_k|\lambda - \lambda^*|\leq 9\mathcal{E}(w_{\lambda^*}) + 675\bar{\lambda^2}||\lambda^*||_0}\{\mathcal{E}(w_{\lambda}) - 15\bar{\lambda}\sum_{k:\lambda^*\ne 0}\sigma_k |\lambda - \lambda^*|\}$,
\item $\{\log(2K)\}^{1/2}n^{-1/2}M\leq 0.13,$
\item $a_n:= \{2\log(2K)\}^{1/2}Mn^{-1/2} + \log(2K)Mn^{-1/2},$
\item For some $t>0$ we are free to set, $\bar{\lambda} := 4a_n(1+ t\{2(1+8a_nM)\}^{1/2} +8t^2a_nM/3)> 6.4 \{\log(2K)\}^{1/2}n^{-1/2},$
\item $s>0$ solves $a_n(1+ s\{2(1+2a_nM)\}^{1/2} + 2t^2a_nM/3) = 9/5,$
\item $\alpha = \exp(-na_n^2s^2)+7\exp(-4na_nt^2).$
\end{enumerate}
\end{assumption}
The technical assumptions are inherited from Theorem 2.2 of \citet{van2008high}. 

The first technical assumption is needed because the quadratic margin
condition $\mathcal{E}(nw_\lambda) \geq cn||w_\lambda-w^0||^2$ only
holds locally for $w_\eta$ within the $\eta$ neighborhood of $w^0$,
$||w_\lambda-w^0||_\infty\leq \eta/n.$ The estimator $\tilde{\lambda}$
strikes the balance between how much excess risk it incurs and how
different it is from the oracle estimator $\lambda^*$ in the $\ell^1$
neighborhood of $\lambda^*$.

The second technical assumption is to ensure the applicability of
Bousquet's inequality to the empirical process induced by $Z$
conditional on $X$. The constant 0.13 is rather arbitrary; it could be
replaced by any constant smaller that $(\sqrt{6}-\sqrt{2})/2$ if other
constants are adjusted accordingly.

The third technical assumption on $a_n$ is due to the usual rate of
decay in probability for Gaussian linear model with orthogonal design,
resulting from a symmetrization inequality and a contraction
inequality.

The fourth technical assumption on $\bar{\lambda}$ is setting a lower
bound for the smoothing parameter. It follows from the Bousquet's
inequality. $t$ is a parameter to be set by users; we need to strike
the balance between small excess risk due to small $t$ and large
confidence in the upper bound for excess risk due to large $t$.

The fifth technical condition on $s$ is due to the contraction inequality for the additional randomness in standard error of covariates $\hat{\sigma}_k$ relative to the true standard deviation $\sigma_k$. 

The sixth technical condition on $\alpha$ defines ``with high probability'' as with probability $1-\alpha$ where $\alpha$ decays exponentially in $n$.

With these assumptions, we have the following theorem.

\begin{theorem}
\label{thm:oracle}
Under Assumption \ref{assumption:oracleassumption} and  Assumption \ref{assumption:oracletechnical}, with probability at least $1-\alpha$, we have
\[\mathcal{E}(w_{\lambda^\dagger}) \leq 3\mathcal{E}(nw_{\lambda^*}) + 225 \bar{\lambda}^2 ||\lambda^*||_0,\]
and 
\[\sum_k \sigma_k| \lambda^\dagger_k-\lambda^*_k|\leq \frac{21}{4}\bar{\lambda}\mathcal{E}(nw_{\lambda^*}) + \frac{ 1575 \bar{\lambda}}{4}||\lambda^*||_0, \]
where $\bar{\lambda} > 0$ is a constant that depends on $K$.
\end{theorem}

\Cref{thm:oracle_words} in \Cref{subsec:oracle} is a consequence of \Cref{thm:oracle} and Assumption \ref{assumption:wtconsistency}.

\Cref{thm:oracle} is a consequence of Theorem 2.2 in \citet{van2008high} where the oracle properties for lasso estimators are established under general convex loss. 
We only need to show that the assumptions for Theorem 5.1 imply the assumptions of Theorem 2.2 in \citet{van2008high} so that their conclusion applies. 

When there are few active covariate balancing constraints,
$||\lambda^*||_0$ will be small. The theorem then says that the excess
risk of minimal weights is of the same order as the oracle estimator.
Therefore, minimal weights mimics the exact balancing weights under a
capped number of effective constraints. In other words, resorting to
approximation in covariate balancing enjoys a similar effect of
capping the number of effective balancing constraints. Hence, minimal
weights is immune to the loss of excessive balancing.

An important practical question is how many covariates we should
balance. Exact balancing weights can only balance a few covariates,
because otherwise the problem does not admit a solution. Minimal
weights relieve this problem: we can balance much more covariates with
$\delta$ appropriately set. This oracle inequality says that we do not
need to worry about excessive balancing. We only need to find a sweet
spot between balancing many covariates loosely and balancing a few
covariates strictly. This amounts to setting $\delta$ appropriately,
which we address in \Cref{sec:practice}.

Below we prove \Cref{thm:oracle}.
\begin{proof}
We only need to show assumptions L, B, and C in Theorem 2.2 of
\citet{van2008high} so that their oracle inequality applies to minimal
weights.

First we show assumption L: the loss function is convex and Lipschitz. The loss function writes $L_{w}(x, z) = -z\cdot (\rho\circ(\rho')^{-1}\circ nw)(x) + ((\rho')^{-1}\circ w)(x)$. Fixing $z$, we have
\[\frac{\partial L_w(x, z) }{\partial w} = \{-z\cdot nw(x)+1\}\{-\frac{\rho''}{n(\rho')^2}(nw(x))\}.\]
This is bounded due to assumptions \ref{assumption:oracleassumption}.1 and \ref{assumption:oracleassumption}.3, implying the Lipschitz property: derivatives of $\rho$ and bounded, $z$ is bounded by [0,1] and $nw$ is bounded due to $n\rho'$ is bounded. The convexity of the loss is shown in Appendix B of \citet{chan2016globally}.

We then show assumption B: the quadratic marginal condition. We compute the second derivative of $EL_w$:
\begin{align*}
\frac{\partial^2 EL_w(X,Z)}{\partial w^2} 
=& -\pi(x)\cdot((n\rho')^{-1})' nw(x)+\{-\pi(x)nw(x)+1\}\{((n\rho')^{-1})''nw(x)\}\\
\geq&\pi(x)\frac{\rho''}{(\rho')^2}(\frac{1}{\pi(x)}) + |\eta|\cdot \{((\rho')^{-1})''nw(x)\}.
\end{align*}
This is lower bounded by a positive constant when $\eta > 0$ is small enough. This is ensured again by Assumption 3.1, in particular the concavity of $\rho.$ The last step is due to a Taylor expansion around $nw(x) = 1/\pi(x)$ in its $\eta$-neighborhood. 

Lastly we show assumption C: $\sum_{k\in\mathcal{K}}\sigma_k|\lambda_k-\tilde{\lambda}_k|\leq |\mathcal{K}| \cdot ||w_\lambda - w_{\tilde{\lambda}}||$. This is again ensured by Assumption \ref{assumption:oracleassumption}.1, in particular the boundedness of the first and second derivative.

The theorem then follows from Theorem 2.2 of \citet{van2008high} where $H= cu^2/2$ and $G =u^2/(2c)$ for some constant $c > 0$ due to the quadratic margin condition.
\end{proof}

\section{Details on Empirical Studies}

\label{sec:empirical_supp}


\subsection{A Remark on the Right Heart Catheterization Study}

A remark on \Cref{table:rhc}(b) is that the optimal error of the
weighting estimator for the average treatment effect on the treated is
sometimes smaller under bad overlap than under good overlap. This may
be counterintuitive, but is a result of the estimand changing under
good and bad overlap when estimating the average treatment effect on
the treated. Specifically, the treated population is different in the
simulated data sets with good and bad overlap, so the estimand is
different. This phenomenon is absent when estimating the average
treatment effect, where the estimand is the same under good and bad
overlap (see \Cref{table:rhc}(a)).


\subsection{The Kang and Schafer Example}

The Kang and Schafer example \citep{kang2007} consists of four
unobserved covariates $U_i
\stackrel{iid}{\sim} N(0, I_4)$, $ i= 1,..., n$. They are used to
generate four covariates $X_i$ that are observed by the investigator:
$X_{i1} = \exp(U_{i1}/2),$ $X_{i2} = U_{i2}/\{{1 + \exp(U_{i1})}\}+ 10,$
$X_{i3} = (U_{i1}U_{i3}+0.6)^3,$ and $X_{i4} = (U_{i2}+U_{i4}+20)^2.$
There is an outcome variable $Y_i$ generated by $Y_i = 210 +
27.4U_{i1} + 13.7 2U_{i2} + 13.7U_{i3} +13.7U_{i4}+\epsilon_i$ where
$\epsilon_i\stackrel{iid}{\sim} N(0, 1)$, and an incomplete outcome
indicator $Z_i$ generated as a Bernoulli random variable with
parameter $p_i = \exp(-U_{i1} - 2U_{i2} - 0.25U_{i3} - 0.1U_{i4})$.
This incomplete outcome indicator denotes whether the outcome is
observed ($Z_i = 1$) or not ($Z_i = 0$).

Using this data generation mechanism, the mean difference of the
observed covariates between the complete and incomplete outcome data
is of $(-0.4, -0.2, 0.1, -0.1)$ standard deviations. We consider this
the ``good overlap'' case. We also consider another case where the
generating mechanism of $p_i$ is slightly different: $p_i =
\exp(-U_{i1} - 0.5U_{i2} - 0.25U_{i3} - 0.1U_{i4}).$ This makes
covariate balance slightly worse, resulting in slightly larger mean
differences of $(-0.3, -0.5, -0.1, -0.4)$ standard deviations. We
consider this the ``bad overlap'' case. 

Tables \ref{table:ks} presents the root mean squared error of the
weighting estimates. Approximate balance outperforms exact balance in
the bad overlap case. The improvement is not as marked as we
documented in the RHC study because the good and bad overlap cases do
not differ much: the mean difference goes from $(-0.4, -0.2, 0.1,
-0.1)$ in the good overlap to $(-0.3, -0.5, -0.1, -0.4)$ in the bad
overlap case. With this relatively small change in covariate balance,
minimal weights immediately outperform the exact balancing weights in
the bad overlap cases. This gives us an understanding of when we
should use minimal weights. We also observe that minimal weights can
sometimes outperform the exact balancing weights in the good overlap
case. 

\begin{table}
\begin{subtable}[h!]{\linewidth}
\centering
\begin{tabular}{lrrrrr}
\toprule
\multicolumn{1}{c}{} & \multicolumn{2}{c}{Good Overlap}
&&\multicolumn{2}{c}{Bad Overlap}\\
\multicolumn{1}{c}{Minimize} & \multicolumn{1}{c}{Exact} &
\multicolumn{1}{c}{Approx.} & &\multicolumn{1}{c}{Exact} &
\multicolumn{1}{c}{Approx.}\\
\midrule Absolute Deviation& \textbf{6.38} & \textbf{6.38} & & 7.83 &
\textbf{7.20}\\ Variance& \textbf{5.71} & 5.79 & & 5.99  &
\textbf{5.65}\\ Negative Entropy& \textbf{5.55} & 5.99& & 5.75 &
\textbf{5.30} \\
\bottomrule
\end{tabular}
\caption{Mean unobserved outcome}
\end{subtable}
\begin{subtable}[t]{\linewidth}
\centering
\begin{tabular}{lrrrrr}
\toprule
\multicolumn{1}{c}{} & \multicolumn{2}{c}{Good Overlap}
&&\multicolumn{2}{c}{Bad Overlap}\\
\multicolumn{1}{c}{Minimize} & \multicolumn{1}{c}{Exact} &
\multicolumn{1}{c}{Approx.} & &\multicolumn{1}{c}{Exact} &
\multicolumn{1}{c}{Approx.}\\
\midrule Absolute Deviation& 6.38 & \textbf{5.01} & & 4.87 &
\textbf{4.80}\\ Variance& \textbf{4.50} & 4.59 & & 4.98  &
\textbf{4.85}\\ Negative Entropy& \textbf{3.70} & 3.85& & 4.97 &
\textbf{4.87} \\
\bottomrule
\end{tabular}
\caption{Mean outcome}
\end{subtable}
\caption{Root mean squared error in the Kang-Schafer study. With bad
overlap, approximate balancing can help reduce the estimation
error.\label{table:ks}}
\end{table}

\begin{figure}[t!]
    \centering
    \begin{subfigure}[t]{0.5\textwidth}
        \centering
        \includegraphics[width=\linewidth]{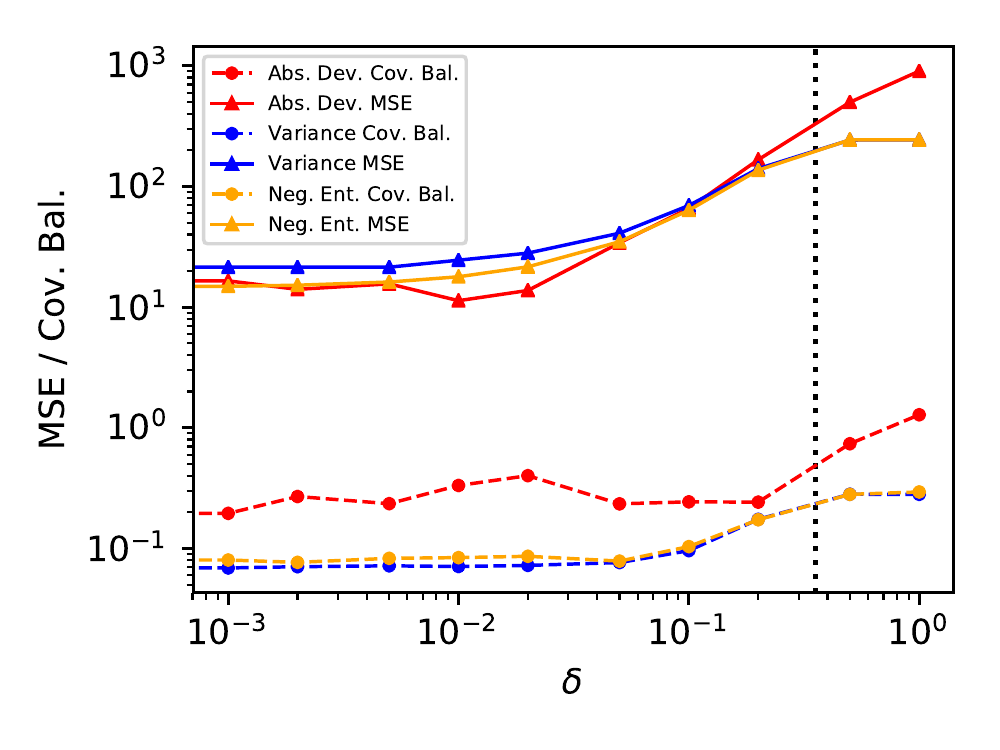}
        \caption{Good overlap}
    \end{subfigure}%
    ~ 
    \begin{subfigure}[t]{0.5\textwidth}
        \centering
        \includegraphics[width=\linewidth]{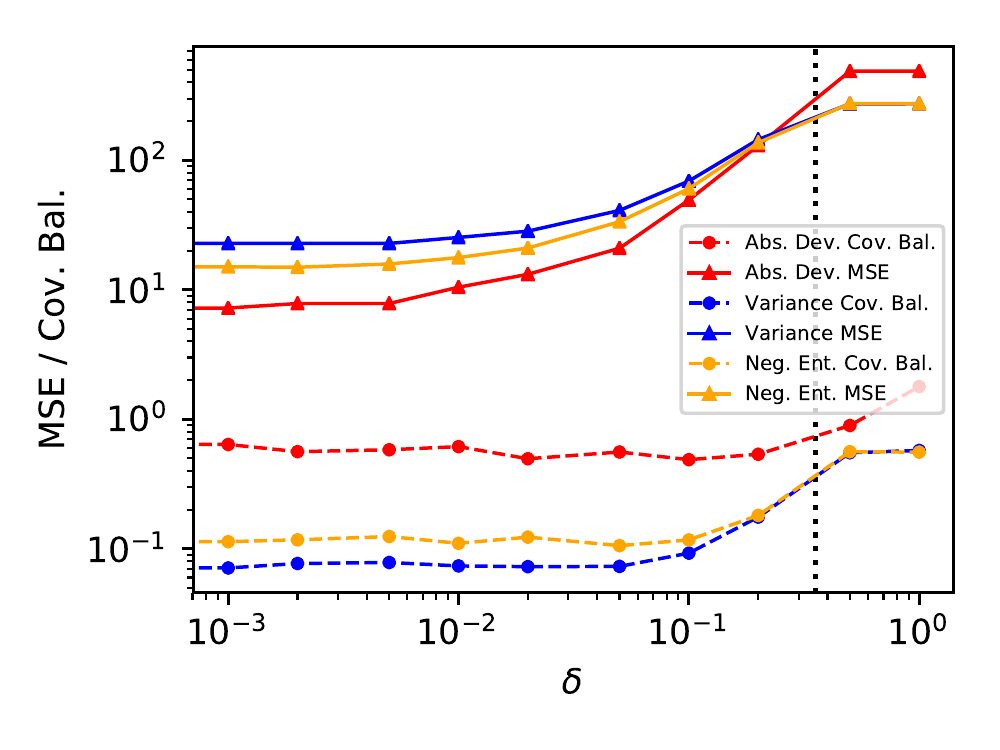}
        \caption{Bad overlap}
    \end{subfigure}
    \caption{Bootstrapped covariate balance $C_S$ and mean squared
    error for different values of $\delta$ for the average treatment
    effect on the treated in the Kang and Shafer study. Using $C_S$ to
    select $\delta$ as in \Cref{alg:tune} coincides with or neighbors
    the optimal $\delta$ with the smallest error. (The horizontal axis
    start from $\delta=0$. The vertical dotted line indicates $\delta
    = K^{-1/2}$, where $K$ is the number of covariates being balanced.
    We recommend not choosing $\delta$'s bigger than $K^{-1/2}$
    because they likely break the assumptions required by the
    asymptotics. )\label{fig:ks_ate_fig}}
\end{figure}

\begin{figure*}[t!]
    \centering
    \begin{subfigure}[t]{0.5\textwidth}
        \centering
        \includegraphics[width=\linewidth]{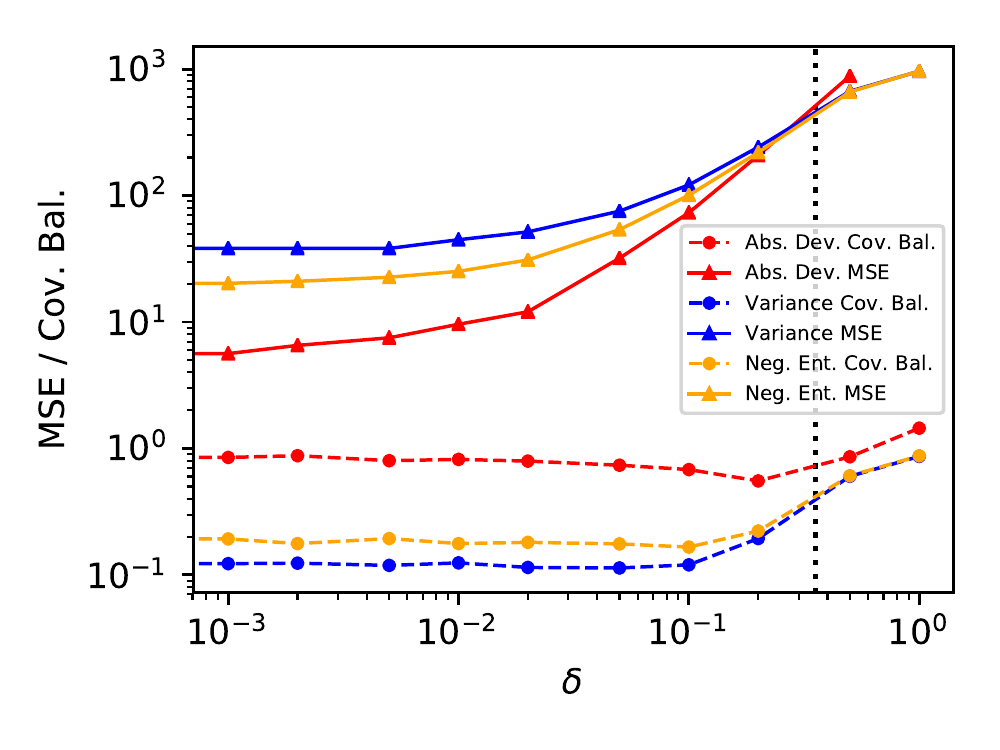}
        \caption{Good overlap}
    \end{subfigure}%
    ~ 
    \begin{subfigure}[t]{0.5\textwidth}
        \centering
        \includegraphics[width=\linewidth]{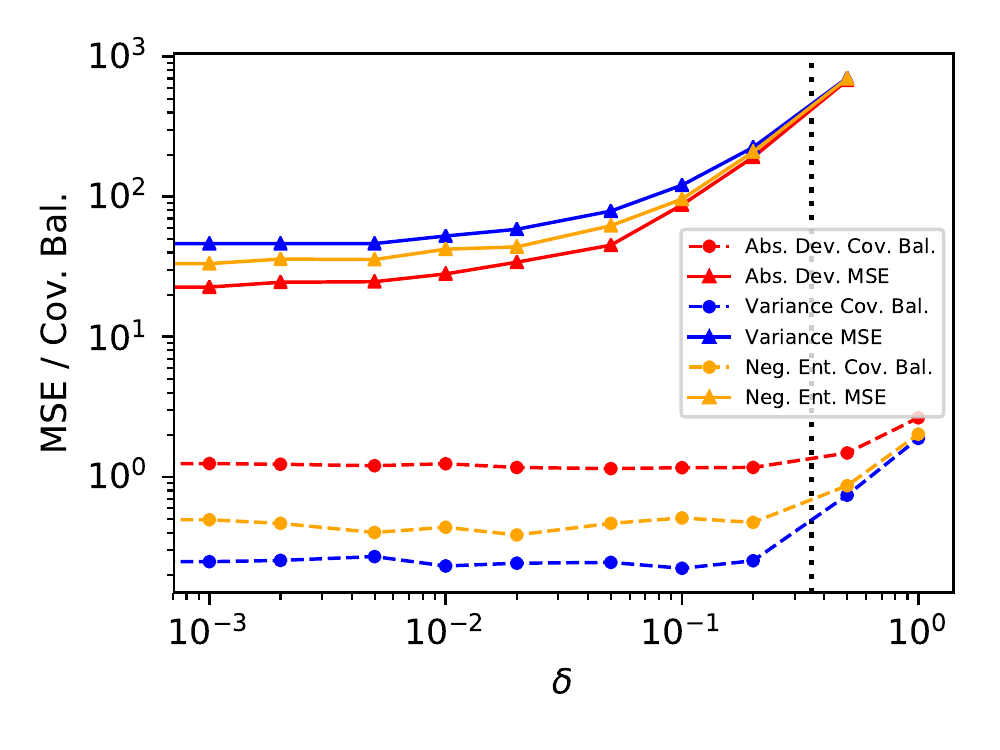}
        \caption{Bad overlap}
    \end{subfigure}
    \caption{Bootstrapped covariate balance $C_S$ and mean squared
    error for different values of $\delta$ for the average treatment
    effect on the treated in the Kang and Shafer study. Using $C_S$ to
    select $\delta$ as in \Cref{alg:tune} coincides with or neighbors
    the optimal $\delta$ with the smallest error. (The horizontal axis
    start from $\delta=0$. The vertical dotted line indicates $\delta
    = K^{-1/2}$, where $K$ is the number of covariates being balanced.
    We recommend not choosing $\delta$'s bigger than $K^{-1/2}$
    because they likely break the assumptions required by the
    asymptotics. )\label{fig:ks_att_fig}}
\end{figure*}

\subsection{The LaLonde Data Set}

We next study the performance of minimal weights in the LaLonde data
set \citep{lalonde1986evaluating}. This data set has two components:
an experimental part from a randomized experiment evaluating a large
scale job training program (the National Supported Work Demonstration,
NSW) on 185 participants; and an observational part, where the
experimental control group from the randomized experiment is replaced
by a control group of 15992 of nonparticipants drawn from the Current
Population Survey (CPS). The experimental part provides a benchmark
for the effect of the job training program to be recovered from
observational part of the data set. This benchmark is \$1794 for the
average treatment effect on the treated with a 95\% confidence
interval of $[551,3038]$.

\Cref{tab:LaLonde} presents the average treatment effect on the
treated estimates and their 95\% confidence intervals using minimal
weights and its exact balancing counterpart. We use $\delta\cdot$sd
for different levels of approximate balancing. Minimal weights
together with the tuning algorithm produces more efficient mean
average treatment effect on the treated  estimates while remaining
close to the experimental target \$1794. The 95\% confidence intervals
all contain the experimental 95\% confidence interval and they become
more efficient as $\delta$ increases. When $\delta$ grows to as large
as 1 sd, the average treatment effect on the treated  estimates starts
to shift away from the target. This is intuitive as overly large
$\delta$ would imply we are no longer balancing the covariates. In
this regard, we conclude minimal weights produce more efficient
average treatment effect on the treated estimates while being faithful
to the truth (experimental target).

\begin{table}[htb]
\centering
\begin{tabular}{llllll}
\toprule
\multicolumn{1}{c}{Minimize} & \multicolumn{1}{c}{Exact} & \multicolumn{1}{c}{Approx.} \\
\midrule
Absolute Deviation& 712 (2602) & \textbf{744 (1257)} \\
Variance& 1668 (1076) & \textbf{1387 (886)} \\
Negative Entropy& 1706 (958) & \textbf{1382 (1078)}\\
\bottomrule
\end{tabular}
\caption{Average treatment effect on the treated estimates in the
Lalonde study. (We present the estimates as mean(sd).)  Minimal
weights produce more efficient estimates while being faithful to the
truth. \label{tab:LaLonde}}
\end{table}


\subsection{The Wong and Chan Simulation}

We finally study the minimal weights in the \citet{wong2018kernel}
simulation. It starts with a ten-dimensional multivariate standard
Gaussian random vector $Z = (Z_1, \ldots, Z_{10})^\top$ for each
observation. Then it generates ten observed covariates $X =
(X_1,\ldots,X_{10})^\top,$ where
\[X_1=\exp(Z_1/2), \]
\[X_2=Z_2/\{1+\exp(Z_1)\}, \]
\[X_3 = (Z_1Z_3/25 + 0.6)^3,\]
\[X_4 = (Z_2+Z_4+20)^2,\]
\[X_j = Z_j, j = 5, \ldots, 10.\]
The propensity score model is 
\[\mathrm{pr}(T=1\,|\, Z) = \exp(-Z_1-0.1Z_4)/\{1+\exp(-Z_1-0.1Z_4)\}.\]
The study considers two outcome regression models. Model A is
\[Y = 210 + (1.5T-0.5)(27.4Z_1+13.7Z_2+13.7Z_3+13.7Z_4)+\epsilon,\]
and model B is 
\[Y = Z_1Z_2^3Z_3^2Z_4 + Z_4|Z_1|^{0.5}+\epsilon,\]
where $\epsilon\sim N(0,1).$ 

We generate a dataset of size $N = 5000$ and study both the average
treatment effect and the average treatment effect on the treated
estimates. (We take the size of the bootstrap samples as 1/10 of the
original sample size. We default to 10 bootstrap samples for covariate
balance evaluation. We balance the first and second moments of the
covariates.)

Tables \ref{table:Wongchan} presents the root mean squared error of
the weighting mean estimates. Approximate balancing with
\Cref{alg:tune} outperforms exact balancing in many cases, especially
in estimating the average treatment effect. The performance is less
stable with the outcome model A, where it could lead to suboptimal
performance. When the treatment indicator interacts with potential
confounders $Z$'s, classical bootstrap agnostic to the treatment
indicator does not serve as a good indicator of downstream estimation
performance. \Cref{fig: WC_A} shows the mean squared error versus
bootstrapped covariate balance plot. The pattern of bootstrapped
covariate balance roughly aligns with the mean squared error. This
implies that selecting $\delta$ with \Cref{alg:tune} (i.e. selecting
according to the bootstrapped covariate balance) could often result in
close-to-optimal error, especially in estimating the average treatment
effect.

\begin{table}
\begin{subtable}[t]{\linewidth}
\centering
\begin{tabular}{lrrrrr}
\toprule
\multicolumn{1}{c}{} & \multicolumn{2}{c}{Outcome model A} &&\multicolumn{2}{c}{Outcome model B}\\
\multicolumn{1}{c}{Minimize} & \multicolumn{1}{c}{Exact} & \multicolumn{1}{c}{Approx.} & &\multicolumn{1}{c}{Exact} & \multicolumn{1}{c}{Approx.}\\
\midrule
Absolute Deviation& 0.67 & \textbf{0.66} & & \textbf{0.26}& \textbf{0.26}\\
Variance& \textbf{0.72} & 0.79 & & 0.26  & \textbf{0.25}\\
Negative Entropy& \textbf{0.78} & 0.89& & \textbf{0.25} & \textbf{0.25} \\
\bottomrule
\end{tabular}
\caption{Average treatment effect on the treated}
\end{subtable}
\begin{subtable}[t]{\linewidth}
\centering
\begin{tabular}{lrrrrr}
\toprule
\multicolumn{1}{c}{} & \multicolumn{2}{c}{Outcome model A} &&\multicolumn{2}{c}{Outcome model B}\\
\multicolumn{1}{c}{Minimize} & \multicolumn{1}{c}{Exact} & \multicolumn{1}{c}{Approx.} & &\multicolumn{1}{c}{Exact} & \multicolumn{1}{c}{Approx.}\\
\midrule
Absolute Deviation& 0.47 & \textbf{0.45} & & \textbf{0.23} & 0.24\\
Variance& 1.35 & \textbf{0.51} & & 0.31  & \textbf{0.21}\\
Negative Entropy& \textbf{0.44} & 0.52& & \textbf{0.21}& \textbf{0.21} \\
\bottomrule
\end{tabular}
\caption{Average treatment effect}
\end{subtable}
\caption{Root mean squared error in the Wong-Chan study. Approximate
balancing often produce similar-or-better quality estimates than exact
balancing.\label{table:Wongchan}}
\end{table}

\begin{figure*}[t!]
    \centering
    \begin{subfigure}[t]{0.5\textwidth}
        \centering
        \includegraphics[width=\linewidth]{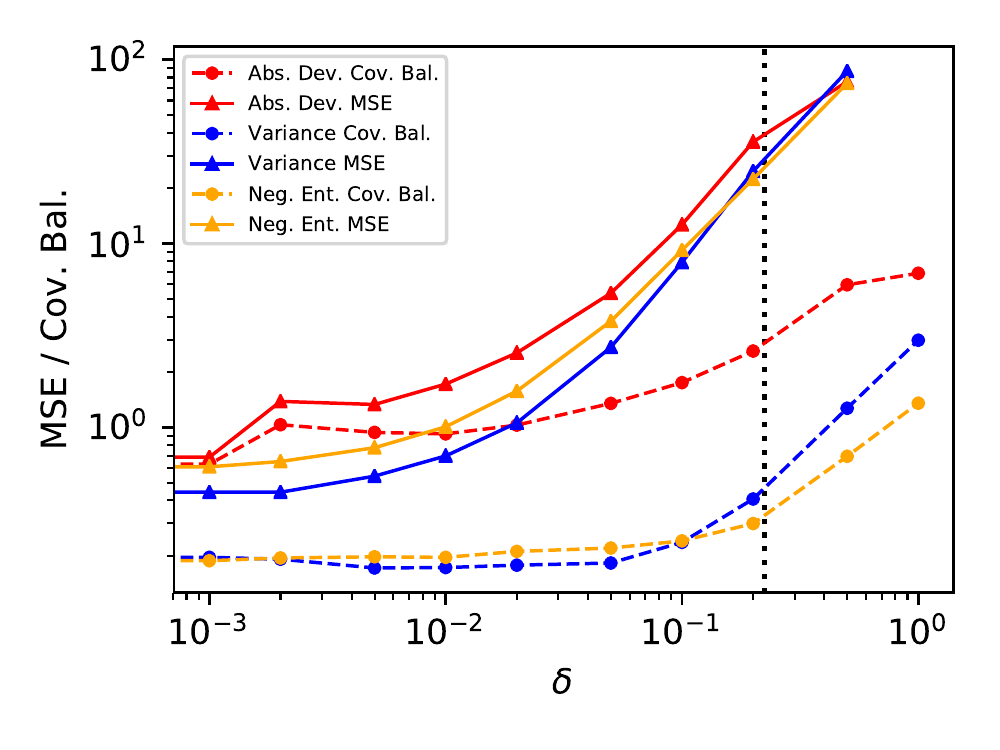}
        \caption{Average treatment effect on the treated}
    \end{subfigure}%
    ~ 
    \begin{subfigure}[t]{0.5\textwidth}
        \centering
        \includegraphics[width=\linewidth]{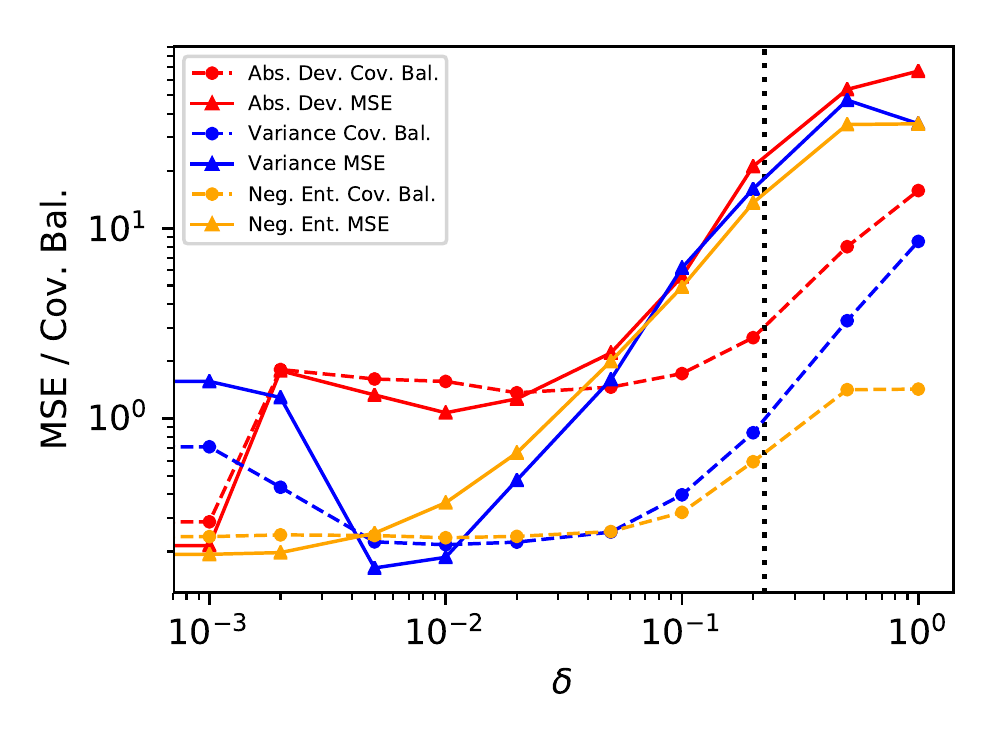}
        \caption{Average treatment effect}
    \end{subfigure}
    \caption{Bootstrapped covariate balance $C_S$ and mean squared
    error for different values of $\delta$ for the average treatment
    effect on the treated in the Wong and Chan study. Using $C_S$ to
    select $\delta$ as in \Cref{alg:tune} coincides with or neighbors
    the optimal $\delta$ with the smallest error, especially in
    estimating the average treatment effect. (The horizontal axis
    start from $\delta=0$. The vertical dotted line indicates $\delta
    = K^{-1/2}$, where $K$ is the number of covariates being balanced.
    We recommend not choosing $\delta$'s bigger than $K^{-1/2}$
    because they likely break the assumptions required by the
    asymptotics. )\label{fig: WC_A}}
\end{figure*}

\end{document}